\documentclass[12pt, draftclsnofoot,onecolumn]{IEEEtran}
\usepackage{mathrsfs}
\usepackage{times, color}
\usepackage{graphicx}
\usepackage{psfrag}
\usepackage{epsfig}
\usepackage{algorithm}
\usepackage{algorithmic}
\usepackage{multicol}
\usepackage{multirow}
\usepackage{setspace}
\usepackage{rotating}
\usepackage{cite}
\usepackage{array}
\usepackage{amssymb}
\usepackage[cmex10]{amsmath}
\usepackage{mdwmath}
\usepackage{mdwtab}
\usepackage[tight,footnotesize]{subfigure}
\usepackage{mdwlist}
\usepackage{epstopdf}
\usepackage{caption}
\usepackage{float}
\usepackage{yhmath}
\usepackage{setspace}
\usepackage{amsthm}
\usepackage{bm}
\usepackage{color}
\usepackage{cite}
\usepackage{graphicx,color}
\usepackage[cmex10]{amsmath}

\usepackage{bm}
\usepackage{amssymb}
\usepackage{amsthm}
\usepackage{algorithm}
\usepackage{algorithmic}
\usepackage{cite}
\usepackage{graphicx,color}
\usepackage[cmex10]{amsmath}
\usepackage{bm}
\usepackage{amssymb}
\usepackage{algorithm}
\usepackage{algorithmic}
\usepackage{setspace}
\usepackage{amsthm}
\usepackage[cmex10]{amsmath}
\usepackage{amssymb}
\usepackage{cite}
\usepackage{graphicx}
\usepackage{array,color}
\usepackage{amsmath}
\usepackage{stfloats}
\usepackage{graphicx}
\usepackage{subfigure}
\usepackage{tabularx}
\usepackage{epsfig,epsf,color,balance,cite}
\usepackage{algorithmic}
\usepackage{algorithm}
\usepackage{epstopdf}
\usepackage{bm}
\usepackage{textcomp}
\usepackage{amsthm}
\newtheorem{theorem}{Theorem}
\newtheorem{corollary}{Corollary}
\newtheorem{lemma}{Lemma}

\newtheorem{problem}{Problem}
\usepackage{graphicx}
\usepackage{setspace}
\begin{document}

\title{Random Caching in Backhaul-Limited Multi-Antenna Networks: Analysis and Area Spectrum Efficiency Optimization}

\author{Sufeng Kuang, \emph{Student Member, IEEE}, Nan Liu, \emph{Member, IEEE}

\thanks
{S.Kuang and N.Liu are with the National Mobile Communications Research Laboratory, Southeast University, Nanjing 210096, China.(Email: sfkuang@seu.edu.cn and nanliu@seu.edu.cn).}}

\maketitle

\begin{abstract}

  Caching at base stations is a promising technology to satisfy the increasing capacity requirements and reduce the backhaul loads in future wireless networks. Careful design of random caching can fully exploit the file popularity and achieve good performance. However, previous works on random caching scheme usually assumed single antenna at BSs and users, which is not the case in practical multi-antenna networks. In this paper, we consider the analysis and optimization in the cache-enabled  multi-antenna networks with limited backhaul.  We first derive a closed-form expression and a simple tight upper bound of the successful transmission probability,  using tools from stochastic geometry and  a gamma approximation. Based on the analytic results, we then consider the area spectrum efficiency maximization by optimizing design parameters, which is a complicated mixed-integer optimization problem. After analyzing the optimal properties, we obtain a local optimal solution with lower complexity. To further simplify the optimization, we then  solve an asymptotic optimization problem in the high user density region, using the upper bound as the objective function.  Numerical simulations show that the asymptotic optimal caching scheme achieves better performance over  existing caching schemes. The analysis and optimization results provide insightful design guidelines for random caching in practical networks.

\end{abstract}

\begin{IEEEkeywords}
Cache, limited backhaul, multi-antenna, Poisson point process, stochastic geometry
\end{IEEEkeywords}

\section {Introduction}\label{Introduction}

The deployment of small base stations (SBSs) or network densification, is proposed as a key method for meeting the tremendous capacity increase in 5G networks \cite{Bhushan14networkdensification}. Dense deployment of small cells can significantly improve the performance of the network by bringing the BSs closer to the users. On the other hand, MIMO technology, especially the deployment of  massive number of antennas at BSs,  is playing an essential role in 5G networks to satisfy the increasing data requirements of the users.  Therefore, the combination of small cells and novel MIMO techniques is inevitable in 5G networks \cite{Jeffrey145G}. However, this approach aggravates the transmission loads of the backhaul links connecting the SBSs and the core networks.

Caching at BSs is a promising method to alleviate the heavy backhaul loads in small cell networks \cite{XWang14Cacheinair}. Therefore, jointly deploying cache and multi antennas at BSs is proposed to achieve 1000x capacity increase for 5G networks \cite{Anliu14MIMOcache,Anliu15MIMOcaching,Meixiatao16cache}. In \cite{Anliu14MIMOcache} and \cite{Anliu15MIMOcaching}, the authors considered the optimization of the cooperative MIMO for video streaming in cache-enabled networks. In \cite{Meixiatao16cache}, the authors considered the optimal multicasting beamforming design for cache-enabled cloud radio access network (C-RAN).  Note that the focuses of the above works were on the parameter optimization for cache-enabled networks under the traditional grid model.

Recently, Poisson point process (PPP) is proposed to model the BS locations  to capture the irregularity and randomness of the small cell networks\cite{Jeffrey11Tractable}.  Based on the random spatial model,  the authors in \cite{Debbah15cache} and \cite{Binxia15Cache} considered storing the most popular files in the cache for small cell networks and heterogeneous cellular networks (HetNets). In \cite{15UCcaching}, the authors considered uniformly caching all the files at BSs, assuming that the file request popularity follows a uniform distribution.  In \cite{Nagananda15IIDcache} , the authors considered storing different files at the cache in an i.i.d. manner, in which each file is selected with its request probability.  Note that in \cite{Debbah15cache,Binxia15Cache,15UCcaching,Nagananda15IIDcache}, the authors did not consider the optimal cache placement, rather, they analyzed the performance under a given cache placement, and thus, the results of the papers might not yield the best performance of the system.

In view of the above problem, random caching strategy is proposed to achieve the optimal performance \cite{Yingcui16Smallnet,Yingcui17Hetnet,Kaibinhuang2016HetCache,Quek17cachecooperation}. In \cite{Yingcui16Smallnet} and \cite{Yingcui17Hetnet}, the authors considered the caching and multicasting in the small cell networks or the HetNets,  assuming random caching at SBSs. In \cite{Kaibinhuang2016HetCache}, the authors considered the  analysis and optimization of the cache-enabled multi-tier HetNets, assuming random caching for all tiers of the BSs. Note that in \cite{Yingcui16Smallnet,Yingcui17Hetnet,Kaibinhuang2016HetCache}, the authors obtained a water-filling-type optimal solutions in some special cases due to the Rayleigh distribution of the fading. In  \cite{Quek17cachecooperation} and \cite{ourwork17WCNC}, the authors considered the helper cooperation for the small cell networks or the HetNets, in which the locations of BSs are modeled as Poisson cluster Process.

However, the works mentioned above considered the random caching in networks equipped with a \emph{single} antenna at BSs and users, which is rarely the case in practical networks.  In \cite{Chenyangyang17cacheHetnet}, the authors considered the cache distribution optimization in  HetNets, where multi-antennas are deployed at the MBSs.  However, the authors considered a special case of zero forcing precoding, i.e., the number of the users is equal to the number of BS antennas, in which the equivalent channel gains from the BSs to its served users follow the Rayleigh distribution. This property does not hold for the general MIMO scenario, where the equivalent channel gains from the BSs to its served users follow the Gamma distribution \cite{Letaief2017MIMOsummary,Alouini16MIMOretransmission}.

The main difficulty of the analysis of the multi-antenna networks stems from the complexity of the random matrix channel \cite{Jeffrey13MIMOHetNet,Changli14SmallCell,Changli16Hetnet}. In \cite{Jeffrey13MIMOHetNet}, the authors utilized the stochastic ordering to compare MIMO techniques in the HetNets,  but the authors did not comprehensively analyze the SINR distribution. In \cite{Changli14SmallCell}, the authors proposed to utilize a Toeplitz matrix representation to obtain a tractable expression of the successful transmission probability in the multi-antenna small cell networks. In \cite{Changli16Hetnet,Changli15IN,Letaief17mmwave}, the authors extended the approach of the Toeplitz matrix representation to analyze the MIMO mutli-user HetNets, MIMO networks with interference nulling and millimeter wave networks with directional antenna arrays.   However, this expression involves the matrix inverse, which is difficult for analysis and  optimization. In \cite{Tianyangbai15mmWave}, a gamma approximation \cite{alzer1997some}  was utilized to facilitate the analysis in the millimeter wave networks.

In this paper, we consider the analysis and optimization of random caching in backhaul-limited multi-antenna networks. Unlike the previous works \cite{Yingcui16Smallnet,Yingcui17Hetnet,Kaibinhuang2016HetCache} focusing on the successful transmission probability of the typical user, we analyze and optimize the area spectrum efficiency, which is a widely-used metric to describe the network capacity.  The optimization is over  \emph{file allocation} strategy and \emph{cache placement} strategy, where a file allocation strategy dictates which file should be stored at the cache of the BSs and which file should be transmitted via the backhaul, and a cache placement strategy is to design the probability vector according to which the files are randomly stored in the cache of the BSs.

First, we derive an exact expression of the successful transmission probability in cache-enabled multi-antenna networks with limited backhaul, using tools from stochastic geometry and a Toeplitz matrix representation. We then utilize a gamma approximation to derive a tight upper bound on the performance metrics to facilitate the parameter design.  The exact expression involves the inverse of a lower triangular Toeplitz matrix and the upper bound is a sum of a series of fractional functions of the caching probability.   These expressions reveal the impacts of the parameters on the performance metrics, i.e., the successful transmission probability and the area spectrum efficiency.  Furthermore, the simple analytical form of the upper bound facilitates the parameter design.

Next, we consider the area spectrum efficiency maximization by jointly optimizing the file allocation and cache placement, which is a very challenging mixed-integer optimization problem.  We first prove that the area spectrum efficiency is an increasing function of the cache placement. Based on this characteristic, we then exploit the properties of the file allocation and obtain a local optimal solution in the general region, in which the user density is moderate. To further reduce the complexity, we then solve an asymptotic optimization problem in the high user density region, using the upper bound as objective function. Interestingly,  we find that the optimal file allocation for the asymptotic optimization is to deliver the most $B$ popular files via the backhaul and store the rest of the files at the cache, where $B$ is the largest number of files the backhaul can deliver at same time. In this way, the users requesting the most $B$ popular files are associated with the nearest BSs, who obtain the most $B$ popular files via the backhaul, and therefore achieve the optimal area spectrum efficiency.

Finally, by numerical simulations, we show that the asymptotic optimal solution with low complexity achieves a significant gain in terms of area spectrum efficiency over previous caching schemes.

\section{System  Model}\label{System Model}

\subsection{Network Model}
We consider a downlink cache-enabled multi-antenna network with limited backhaul, as shown in Fig. \ref{Networkmodel}, where BSs, equipped with $N$ antennas, are distributed according to a homogeneous Poisson point process (PPP) $\Phi_b$ with density $\lambda_b$. The locations of the single-antenna users are distributed as an independent homogeneous PPP $\Phi_u$ with density $\lambda_u$. According to Slivnyak's theorem \cite{haenggi2012stochastic}, we analyze the performance of the typical user who is located at the origin without loss of generality. All BSs are operating on the
same frequency band and the users suffer intercell interference from other BSs. We assume that all the BSs are active due to high user density.

We assume that each user requests a certain file from a finite content library which contains $F$ files. Let $\mathcal{F} = \{1,2,3,\cdots,F \}$ denotes the set of the files in the network. The popularity of the requested files is known a priori and is modeled as a Zipf distribution \cite{LBreslau99CacheZipf}
\begin{align}\label{Zipfdistribution}
q_f=\frac {f^{-\gamma} } {\sum_{i=1}^F i^{-\gamma}   },
\end{align}
where $q_f$ is the probability that a user requests file $f$  and $\gamma$ is the shape parameter of the Zipf distribution. We assume that all the files have same size and the size of a file is normalized to $1$ for simplicity.

Each BS is equipped with a cache with $C$ segments and the cache can store at most $C$ different files out of the content library. For the files which are not stored in the cache, the BSs can fetch them from the core network via backhaul links which can transmit at most $B$ files at same time. We refer to $C$ as  \emph{cache size} and $B$ as \emph{backhaul capability}.   We assume  $B+C \leq F$ to illustrate the resource limitation.

\begin{figure}
\begin{minipage}[t]{0.2\linewidth}
\centering
    \includegraphics[height=2in,width=1.5in]{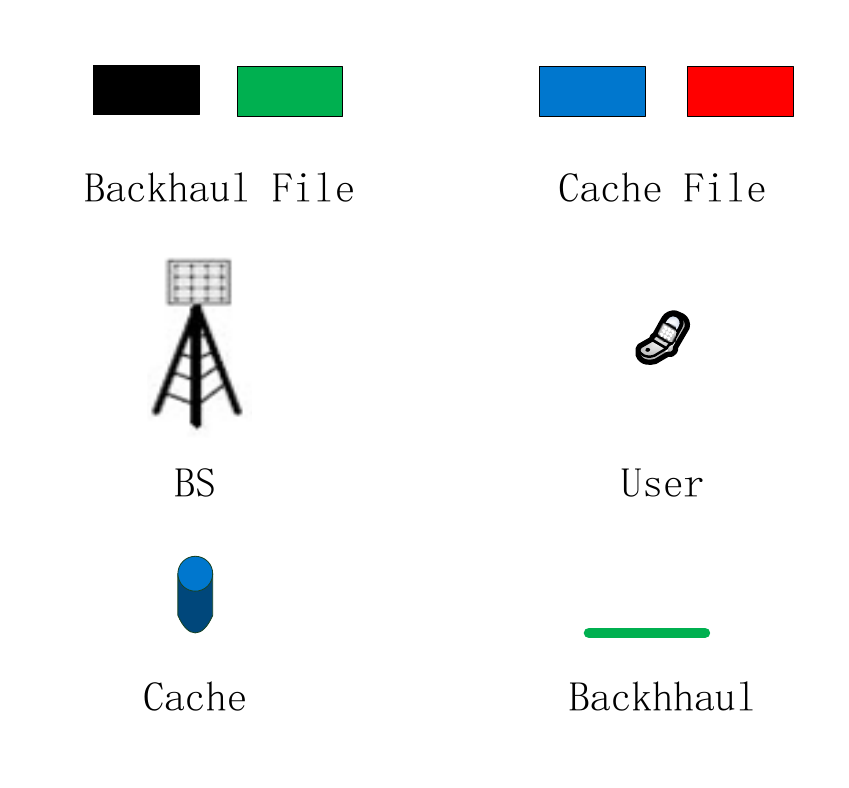}
\end{minipage}
\begin{minipage}[t]{0.43\linewidth}
\centering
    \includegraphics[height=2in,width=2.5in]{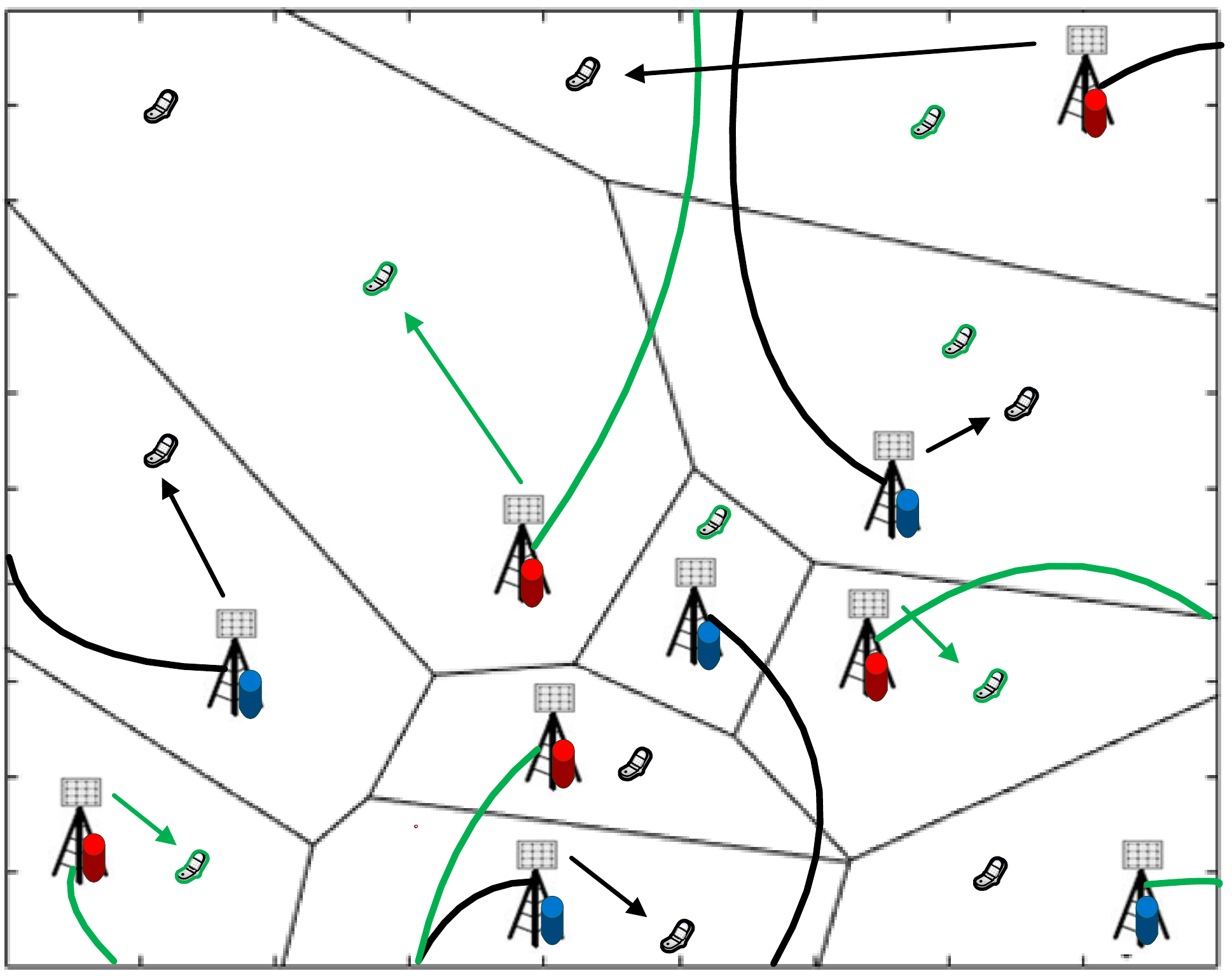}
\caption*{Distance-based Association.}
\end{minipage}
\begin{minipage}[t]{0.35\linewidth}
\centering
    \includegraphics[height=2in,width=2.5in]{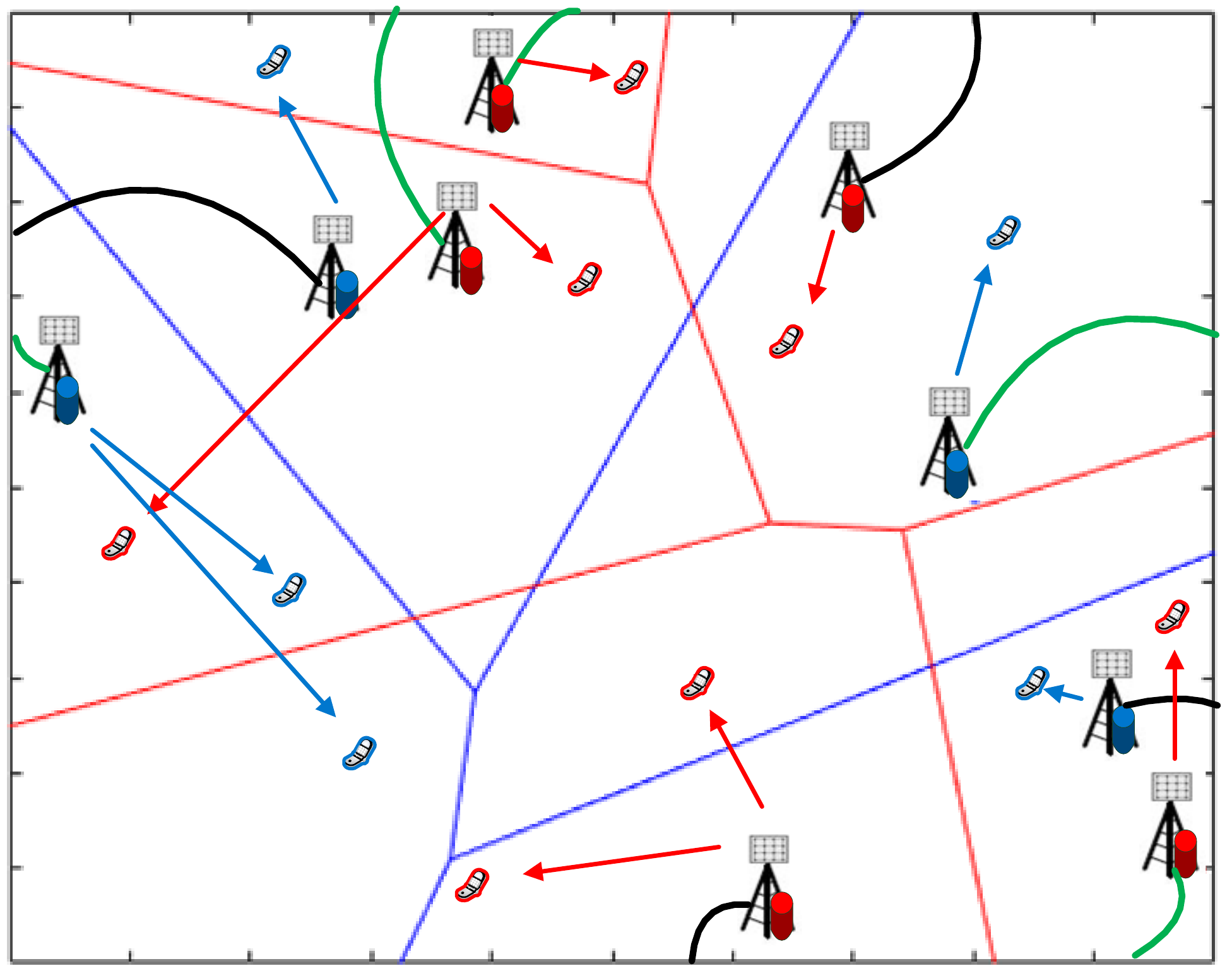}
\caption*{Content-centric Association.}
\end{minipage}
\caption{Illustration of the system model. For the backhaul file, the distance-based association scheme is adopted. For the cached file, the content-centric association scheme is adopted. }
\label{Networkmodel}
\end{figure}

\subsection{Caching and Backhaul Delivery}

The set of all the files is partitioned into two disjoint sets, and we define the set of files stored in the cache of all the BSs and the set of files not stored in any of the BSs, i.e., the files in it must be transmitted via the backhaul, as \emph{cached file set}  and \emph{backhaul file set}, which are denoted as $\mathcal{F}_c$ and $\mathcal{F}_b$, respectively. The number of files in $\mathcal{F}_x$ is $F_x$, $x=c,b$. Since $\mathcal{F}_c$ and $\mathcal{F}_b$ form a partition of $\mathcal{F}$, we have
\begin{align}\label{fileconstraint}
\mathcal{F}_c \bigcup \mathcal{F}_b =\mathcal{F}, ~~\mathcal{F}_c  \bigcap \mathcal{F}_b =\emptyset.
\end{align}
We define the process of designing $\mathcal{F}_c$ and $\mathcal{F}_b$ as \emph{file allocation}.

Each BS can cache $C$ different files from a total $F_c$ files via the \emph{random caching} scheme, in which the BS stores a certain file $i \in \mathcal{F}_c$ randomly with probability $t_i$. Let $\mathbf{t}=(t_i)_{i \in \mathcal{F}_c}$ denotes the caching distribution of all the files. Then we have the following constraints
\begin{align}
&t_i \in [0,1], ~\forall i \in \mathcal{F}_c, \label{T01} \\
&\sum_{i \in \mathcal{F}_c} t_i \leq C.  \label{Ttotal}
\end{align}
We refer to the specification of $\mathbf{t}$ as \emph{cache placement}.

For further analysis, we define the BSs that have the cached file $f \in \mathcal{F}_c$ in the cache as  the \emph{f-cached BSs}. According to the thinning theory of the PPP, the density of the $f$-cached BSs is $\lambda_b^f=t_f \lambda_b$. We denote the set of the $f$-cached BSs as $\Phi_b^f$ and the set of the remaining BSs that do not have file $f$ in their cache as  $\Phi_b^{-f}$.  When a user requests the file $f$ out of the cached file set, the user is associated with the nearest $f$-cached BSs which will provide the required file $f$ from its cache. We refer to the BS associated with the typical user as the \emph{tagged BS} and this association scheme is called \emph{content-centric association} scheme. Content-centric association scheme is different from the traditional \emph{distance-based association} scheme, where the user is associated with the nearest BS. We denote the tagged BS as BS $1$.

When a user requests a file out of the backhaul file set, the distance-based association is adopted and the user is associated with the nearest BS. We define the set of \emph{backhaul requested files} as the set of the backhaul files requested by the users of BS $i$, which is denoted as $\mathcal{F}_{b,i}^r \subseteq \mathcal{F}_b$. The number of the backhaul request file is denoted as $F_{b,i}^r$. If BS $i$ needs to transmit less than $B$ backhaul files via backhaul, i.e., $F_{b,i}^r \leq B$, then BS $i$ gets all $F_{b,i}^r$ files from the backhaul and transmits all of them to the designated users; otherwise,  BS $i$ will randomly select $B$ different files from $\mathcal{F}_{b,i}^r$ according to the uniform distribution, where these files will be transmitted from the backhaul links and then passed on to the designated users.

\section{Performance Metric and Problem Formulation}\label{OptimalDesign}
\subsection{Performance Metric}
In this part, we define the successful transmission probability (STP) and the area spectrum efficiency (ASE) of the typical user when the single-user maximal ratio combination (MRT) beamforming is adopted. We consider the single-user MRT beamforming due to the low complexity of the beamforming design, which is suitable for the scenario when a large number of antennas are deployed at the BSs. For other MIMO scenarios, the equivalent channel gain follows the gamma distribution with different parameters \cite{Letaief2017MIMOsummary,Alouini16MIMOretransmission}.  Therefore, the method of the analysis here can be extended to the cache-enabled MIMO networks   with different precoding and combining strategies.   We assume that the transmitter can get the perfect channel state infomation (CSI) through the feedback from the users. The BSs do not have the CSI of the other cells due to the high BS density.

We refer to the typical  user as user $0$ and it is served by the tagged BS located at $\mathbf{x}_1$.
Due to the assumption of single-user MRT, each BS serves \emph{one} user per resource block (RB). Hence, the received signal of the typical user on its resource block is
\begin{align}\label{signalmodel}
y_0^f=\| \mathbf{x}_1 \|^{-\frac {\beta} {2}} \mathbf{h}_{0,1}^{*} \mathbf{w}_1 s_1 + \sum_{i \in \left \{ \Phi_b \setminus 1 \right \} } \| \mathbf{x}_i \|^{-\frac {\beta} {2}} \mathbf{h}_{0,i}^{*} \mathbf{w}_i s_i + n_0,
\end{align}
 where $\mathbf{w}_i$ is the beamforming vector of BS $i$ to its served user,  $f$ is the file requested by the typical user,
$\mathbf{x}_i$ is the location of BS $i$, $\mathbf{h}_{0,i} \in \mathbb{C}^{N\times 1}$ is the channel coefficient vector from BS $i$ to the typical user, $s_i \in \mathbb{C}^{1 \times 1}$ is the transmitting message of BS $i$ and $n_0$ is the additive white Gaussian noise (AWGN) at the receiver. We assume that $\mathbb{E}[{s_i}^* s_i]=P$ for any $i$ and $P$ is the transmit power of the BS. The elements of the channel coefficient vector $\mathbf{h}_{0,i}$ are independent and identical complex Gaussian random variables, i.e., $\mathcal{CN} (0,1)$. $\beta>2$ is the pathloss exponent.

For single-user MRT,  to maximize the channel gain from the BS to its served user, the beamforming at BS $i$ is  $\mathbf{w}_i=\frac {\mathbf{h}_i} {\|\mathbf{h}_i\|}$ \cite{Jeffrey08Gammadistribution}, where $\mathbf{h}_i \in  \mathbb{C}^{N \times 1}$ is the channel coefficient from BS $i$ to its served user. Thus the SIR of the typical user requiring file $f$, whether $f$ is in the cached file set or the backhaul file set, is given by
\begin{align}\label{MRTSIR}
\textrm{SIR}_f= \frac {P \| \mathbf{x}_1 \|^{-\beta} g_1 } {\sum_{i \in \left \{ \Phi_b \setminus 1 \right \} }  P \| \mathbf{x}_i \|^{-\beta} g_i },
\end{align}
where $g_{i}= \frac {\| \mathbf{h}_{0,i}^*\mathbf{h}_i \|^2} {\|\mathbf{h}_i\|^2}$ is the equivalent channel gain (including channel coefficient and the beamforming) from BS $i$ to the typical user.  It is shown  that the  equivalent channel gain from the tagged BS to its served user, i.e., $g_{1} \sim \textrm{Gamma}(N,1)$ and  $g_{i} \sim \textrm{Exp}(1), \forall i>1$ \cite{Jeffrey08Gammadistribution}. In this paper, we consider SIR for performance analysis rather than SINR due to the dense deployments of the SBSs. In simulations, we include the noise to illustrate that in dense networks  the consideration of the SIR achieves nearly the same performance as that of the SINR.

For the single-user MRT, the successful transmission probability (STP) of the typical user is defined as the probability that the SIR is larger than a threshold, i.e.,
\begin{align}\label{MRTfileSP}
P_{\textrm{s}}(\mathcal{F}_c,\mathbf{t})=\sum_{f \in \mathcal{F}_c} q_f \mathbb{P} \left(\textrm{SIR}_f>\tau \right) +\sum_{f \in \mathcal{F}_b} q_f \mathbb{P} \left( \textrm{SIR}_f>\tau, f \textrm{ transmitted through backhaul} \right),
\end{align}
where $\tau$ is the SIR threshold. As mentioned before, $f$ is transmitted by the backhaul if the number of the backhaul requested files of the tagged BS $F_{b,1}^r$ is no more than the backhaul capability $B$, or if it is been chosen to be transmitted according to the uniform distribution in the event where $F_{b,1}^r >B$. Note that the STP  is related to the random variables $F_{b,1}^r$ and $\textrm{SIR}_f$ .

We use the area spectrum efficiency (ASE) as the metric to describe the average spectrum efficiency per area. The ASE of the single-user MRT is defined as \cite{Baccelli06multihop,Quek12Throughput}
\begin{align}\label{ASEMRT}
R_{}(\mathcal{F}_c,\mathbf{t})=\lambda_b  P_{\textrm{s}}(\mathcal{F}_c,\mathbf{t}) \log_2(1+\tau).
\end{align}
where the unit is bit/s/Hz/$\textrm{km}^2$. Note that the ASE reveals the relationship between the BS density and the network capacity.

\subsection{Problem Formulation}
Under given backhaul capability $B$ and cache size $C$, the caching strategy, i.e., the file allocation strategy and the cache placement strategy, fundamentally affects the ASE. We study the problem of maximizing the ASE via a careful design of file allocation $\mathcal{F}_c$ and cache placement $\mathbf{t}$  as follows
\begin{align}\label{problemformulation}
 &\max_{\mathcal{F}_c,\mathbf{t}} ~~R_{}(\mathcal{F}_c,\mathbf{t})  \\
 &s.t ~~~(\ref{fileconstraint}),(\ref{T01}),(\ref{Ttotal}). \notag
\end{align}

We will derive the expressions of the STP and the ASE  in Section \ref{Performance Analysis} and solve the ASE maximization problem, i.e., the problem in (\ref{problemformulation}), in Section \ref{ASEoptimization}.

\section{Performance Analysis}\label{Performance Analysis}
In this section, we first derive an exact expression of the STP and ASE under given file allocation and cache placement strategy, i.e., under given $\mathcal{F}_c$ and $\mathbf{t}$.  Then we utilize a gamma distribution approximation to obtain a simpler upper bound of the STP and ASE.

\subsection{Exact Expression }

In this part, we derive an exact expression of the STP and the ASE using tools from stochastic geometry. In general, the STP is related to the number of the backhaul request file of the tagged BS, i.e., $F_{b,1}^r$. Therefore, to obtain the STP, we first calculate the probability mass function (PMF) of $F_{b,1}^r$.

\begin{lemma}\label{fbrspmf}
(pmf of $F_{b,1}^r$) When $f \in \mathcal{F}_b$ is requested by the tagged BS, the pmf of $F_{b,1}^r$ is given by
\begin{align}\label{pmfkbr}
\mathbb{P}_f^{\mathcal{F}_b}\left(F_{b,1}^r=k \right )= g \left( \{ \mathcal{F}_b \setminus f \},k-1 \right), \quad k \in\{1,2,\cdots, F_b \},
\end{align}
where $g(\mathcal{B},k) $  is given by
\begin{align}\label{fcfbexpression}
g \left(\mathcal{B},k \right) \overset{\bigtriangleup} = \sum_{\mathcal{Y} \in \left \{\mathcal{X} \subseteq \mathcal{B} : |\mathcal{X} |=k \right \} }
\prod_{i \in \mathcal{Y} } \left(1- {\left(1+ \frac {q_i \lambda_u} {3.5  \lambda_b} \right)}^{-4.5} \right)
\prod_{i \in \mathcal{B} \setminus \mathcal{Y} } {\left(1+ \frac {q_i  \lambda_u} {3.5 \lambda_b} \right)}^{-4.5}.
\end{align}
\end{lemma}

\emph{Proof:} See Appendix \ref{Proofloadspmf}.

Based on Lemma \ref{fbrspmf}, we have the following corollary.

\begin{corollary}\label{asymloads}
(The pmf of $F_{b,1}^r$ when $\lambda_u \rightarrow \infty$). When $\lambda_u \rightarrow \infty$, the pmf of $F_{b,1}^r$ is
\begin{align}
\lim_{\lambda_u \rightarrow \infty}\mathbb{P}_f^{\mathcal{F}_b}\left(F_{b,1}^r=k \right )= \left \{ \begin{aligned}  &0,~~ k=1,2,\cdots,F_b-1  \\
&1,~~k=F_b
\end{aligned} \right..
\end{align}
\end{corollary}

Corollary \ref{asymloads} interprets that  $F_{b,1}^r$ converges to constant $F_b$ in distribution as $\lambda_u \rightarrow \infty$. The asymptotic result is consistent with the fact that when the user density is high, each BS will have many users connected to it, and thus, each BS will require all the backhaul files.

We then calculate the STP,  which is defined in (\ref{MRTfileSP}). Base on Lemma \ref{fbrspmf}, we can rewrite the STP as a combination of the STP conditioned on the given $F_{b,1}^r$. Therefore, we can obtain the STP in Theorem \ref{MRTsuccessP}.

\begin{theorem}\label{MRTsuccessP}
(STP) The STP is given by
\begin{align}\label{MRTsuccessall}
P_{\textrm{s}}(\mathcal{F}_c,\mathbf{t})&= \sum_{f \in \mathcal{F}_c} q_f  P_{\textrm{s}}^{f,c} \left(t_f \right) + \sum_{f \in \mathcal{F} \setminus \mathcal{F}_c} q_f  \sum_{k=1}^{F_b} \mathbb{P}_f^{\mathcal{F}_b}\left(F_{b,1}^r=k \right ) \frac{B} {\max\left(k,B\right)} P_{\textrm{s}}^{b},
\end{align}
where  $\mathbb{P}_f^{\mathcal{F}_b}\left(F_{b,1}^r=k \right )$ is given in (\ref{pmfkbr}), $q_f$ is given in (\ref{Zipfdistribution}), $P_{\textrm{s}}^{f,c} \left( t_f \right)$ and $P_{\textrm{s}}^{b}$ are the STPs of the cached file $f \in \mathcal{F}_c$ and the backhaul file $f \in \mathcal{F} \setminus \mathcal{F}_c$, which are given by
\begin{align}
&P_{\textrm{s}}^{f,c}(t_f)=\frac {t_f} {t_f+ l_{0}^{c,f}} \left \| \left[ \mathbf{I}- \left(\frac {\tau^{2 \backslash \beta}} {t_f+ l_{0}^{c,f}}\right)\  \mathbf{D}_{}^{c,f} \right]^{-1} \right \|_1, \label{PsuccessfC} \\
&P_{\textrm{s}}^{b}=\frac {1} {1+ l_{0}^{b,f}} \left \| \left[ \mathbf{I}- \left(\frac {\tau^{2 \backslash \beta}} {1+ l_{0}^{b,f}}\right)\  \mathbf{D}_{}^{b,f} \right]^{-1} \right \|_1, \label{PsuccessfB}
\end{align}
where $ \|\cdot \|_1$is the $l_1$ induced matrix norm (i.e, $|| \mathbf{B} ||_1= \max_{1 \leq j \leq n} \sum_{i=1}^m |b_{ij}|, ~\mathbf{B} \in \mathbb{R}^{m \times n}$),  $\mathbf{I}$ is an $N \times N$ identity matrix,
$\mathbf{D}_{}^{c,f}$ and $\mathbf{D}_{}^{b,f}$ are $N \times N$ Toeplitz matrices of the cached file $f \in \mathcal{F}_c$ and the backhaul file $f \in \mathcal{F} \setminus \mathcal{F}_c$, which are given by
 \begin{align}\label{matrixT}
\mathbf{D}_{}^{n,f}=\begin{bmatrix}
0 &  \\
l_{1}^{n,f} &0 \\
l_{2}^{n,f} & l_{1}^{n,f} &0  \\
\vdots &\vdots & &\ddots \\
l_{N}^{n,f}& l_{N-1,{}}^{n,f} &\cdots &l_{1}^{n,f} & 0
\end{bmatrix},~n \in \{c,b\},
\end{align}
where $l_{0}^{c,f}$, $l_{0}^{b,f}$, $l_{i}^{c,f}$ and $l_{i}^{b,f}$ are given by
\begin{align}
l_{0}^{c,f}&= t_f \frac{  2\tau} {\beta-2} {}_2 F_1 \left[1,1-\frac{2} {\beta};2-\frac{2} {\beta};-\tau\ \right]  +(1-t_f) \frac {2 \pi} {\beta} \csc \left(\frac{2 \pi} {\beta} \right) \tau^{2 \backslash \beta}, \\
l_{0}^{b,f}&=\frac{  2\tau} {\beta-2} {}_2 F_1 \left[1,1-\frac{2} {\beta};2-\frac{2} {\beta};-\tau\ \right] , \\
l_{i}^{c,f}&=(1-t_f)\frac{2} {\beta} B(\frac{2} {\beta}+1,i-\frac{2} {\beta})  + t_f\frac{2 \tau^{i-2 \backslash \beta}} {i \beta-2} {}_2 F_1 \left[i+1,i-\frac{2} {\beta};i+1-\frac{2} {\beta};-\tau\ \right], \forall i \geq 1, \\
l_{i}^{b,f}&=\frac{2 \tau^{i-2 \backslash \beta}} {i \beta-2} {}_2 F_1 \left[i+1,i-\frac{2} {\beta};i+1-\frac{2} {\beta};-\tau\ \right], \forall i \geq 1.
\end{align}
Here,  ${}_2 F_1(\cdot)$ is the Gauss hypergeometric function and $B(\cdot)$ is the Beta function.
\end{theorem}

\begin{proof}
 Considering the equivalent channel gain $g_{1} \sim \textrm{Gamma}(N,1)$, the STP is a complex $n$-th derivative of the interference Laplace transform  $\mathcal{L}_{I_{}} \left(s \right)$ \cite{Jeffrey13MIMOHetNet}.  Utilizing the approach in \cite{Changli14SmallCell}, we obtain the expressions of the STPs in lower triangular Toeplitz matrix representation as (\ref{PsuccessfC}) and (\ref{PsuccessfB}). For details, please see Appendix \ref{ProofMRTperformance}.
 \end{proof}

According to (\ref{ASEMRT}), we then obtain the ASE under given $\mathcal{F}_c$ and $\mathbf{t}$
\begin{align}\label{ASEMRTresult}
R_{}(\mathcal{F}_c,\mathbf{t})=\lambda_b  P_{\textrm{s}}(\mathcal{F}_c,\mathbf{t}) \log_2(1+\tau).
\end{align}

For backhaul-limited multi-antenna networks, the change of the BS density $\lambda_b$ and user density $\lambda_u$ influence $P_{\textrm{s}}(\mathcal{F}_c,\mathbf{t})$ via the pmf  of the backhaul request file, i.e., $ \mathbb{P} \left(F_{b,1}^r=k \right)$. However, when   $\lambda_u$ approaches infinity and $\lambda_b$ remains finite, since all the backhaul files are requested by the tagged BS (shown in Corollary \ref{asymloads}), $\mathbb{P}_f^{\mathcal{F}_b}\left(F_{b,1}^r=k \right)$ is no longer related to $\lambda_b$. Therefore, when the user density is high and the design parameter $\mathcal{F}_c$ and $\mathbf{t}$  are given,  deploying more BSs will always increase the ASE.

From Theorem \ref{MRTsuccessP} and the definition of the ASE, we can derive a tractable expression of the ASE for $N=1$, i.e., the backhaul-limited single-antenna networks. The ASE of the backhaul-limited single-antenna networks is given in the following corollary.
\begin{corollary}\label{STPSISO}
(ASE of Single-Antenna Networks) The ASE of the cache-enabled single-antenna networks with limited backhaul is given by
\begin{align}\label{PSISIO}
R_{\textrm{SA}} (\mathcal{F}_c,\mathbf{t})=\lambda_b \log_2(1+\tau) \left( \sum_{f \in \mathcal{F}_c} \frac { q_f t_f} {\zeta_1 (\tau) t_f +\zeta_2}+ \sum_{f \in \mathcal{F} \setminus \mathcal{F}_c}  \sum_{k=1}^{F_b} \frac {\mathbb{P}_f^{\mathcal{F}_b}\left(F_{b,1}^r=k \right ) q_f  B } {\max \left( k,B \right) \left(\zeta_1 (\tau) +\zeta_2 \right)} \right),
\end{align}
where $\zeta_1 (\tau)$ and $\zeta_2 (\tau)$ satisfy
\begin{align}
&\zeta_1 (\tau)=   1+\frac{  2\tau} {\beta-2} {}_2 F_1 \left[1,1-\frac{2} {\beta};2-\frac{2} {\beta};-\tau\ \right] -\frac {2 \pi} {\beta} \csc \left(\frac{2 \pi} {\beta} \right) \tau^{2 \backslash \beta}, \label{expressionzeta1} \\
&\zeta_2 (\tau)= \frac {2 \pi} {\beta} \csc \left(\frac{2 \pi} {\beta} \right) \tau^{2 \backslash \beta}  \label{expressionzeta2}.
\end{align}
\end{corollary}

The proof of Corollary \ref{STPSISO} is similar to the proof of Theorem \ref{MRTsuccessP} except that the equivalent channel gain $g_{1} \sim  \textrm{Exp}(1)$.

\subsection{Upper Bound  and Asymptotic Analysis}

In this part, we first derive an  upper bound of the STP under given $\mathcal{F}_c$ and $\mathbf{t}$.  We then give the asymptotic analytic results in high user density region. First, we introduce a useful lemma to present a lower bound of the gamma distribution.

\begin{lemma}\label{Gammaupperbound}
\cite{alzer1997some}: Let $g$ be a gamma random variable follows $\text{Gamma}(M,1)$. The probability $\mathbb{P} (g<\tau)$ can be lower bounded by
\begin{align}
\mathbb{P} (g<\tau)>{\left[1-e^{-a \tau} \right]}^M,
\end{align}
where $\alpha={(M!)}^{-\frac {1} {M}}$.
\end{lemma}
Utilizing the above lemma, we then obtain the upper bound of the STP as follows.

\begin{theorem}\label{upperboundMRT}
(Upper Bound of STP) The upper bound of the STP is given by
\begin{align}\label{MRTsuccessallupper}
P_{\textrm{s}}^u (\mathcal{F}_c,\mathbf{t})&= \sum_{f \in \mathcal{F}_c} q_f  P_{\textrm{s}}^{u,f,c} (t_f) + \sum_{f \in \mathcal{F} \setminus \mathcal{F}_c} q_f  \sum_{k=1}^{F_b} \mathbb{P}_f^{\mathcal{F}_b}\left(F_{b,1}^r=k \right ) \frac{B} {\max\left(k,B\right)} P_{\textrm{s}}^{u,b},
\end{align}
where $P_{\textrm{s}}^{u,f,c}(t_f)$ and $P_{\textrm{s}}^{u,b}$ are the upper bounds of the STPs of the cached file $f \in \mathcal{F}_c$ and the backhaul file $f \in \mathcal{F} \setminus \mathcal{F}_c$, which are given by
\begin{align}
&P_{\textrm{s}}^{u,f,c} (t_f)= \sum_{i=1}^{N} \frac { (-1)^{i+1} \binom {N} {i} t_f} {(\theta_A \left(i \right) t_f +\theta_C \left(i \right))}, \label{uppercachefile} \\
&P_{\textrm{s}}^{u,b}=\sum_{i=1}^{N} \frac {(-1)^{i+1} \binom {N} {i} } {\left(\theta_A \left(i \right)+\theta_C \left(i \right)\right)} , \label{upperbackhaulfile}
\end{align}
where
\begin{align}
\theta_A \left(i \right)&=1+\frac{  2\tau} {\beta-2} {}_2 F_1 \left[1,1-\frac{2} {\beta};2-\frac{2} {\beta};-i \alpha \tau\ \right] -\frac {2 \pi} {\beta} \csc \left(\frac{2 \pi} {\beta} \right) (i \alpha \tau)^{2 \backslash \beta}  \\
\theta_C \left(i \right)&=\frac {2 \pi} {\beta} \csc \left(\frac{2 \pi} {\beta} \right) (i \alpha \tau)^{2 \backslash \beta}.
  \end{align}
  Here, $\alpha={(N!)}^{-\frac {1} {N}}$ is a constant related to the number of BS antennas $N$.
\end{theorem}
\begin{proof}
The STP of the cached file $f \in \mathcal{F}_c$ is
\begin{align}\label{upperboundinequality}
&P_{\textrm{s}}^{f,c}(t_f) =\mathbb{P}\left(g_{\textrm{1}}>\left( \tau I_{}{\| \mathbf{x}_1\|}^{\beta} \right) \right) \notag \\
& \overset{(a)} {\leq} 1-\mathbb{E}_{I_{}{\| \mathbf{x}_1\|}^{\beta}} \left[ \left( 1-\exp  \left( -\alpha \tau I_{}{\| \mathbf{x}_1\|}^{\beta} \right) \right)^{N} \right] \notag \\
& =\sum_{i=1}^{N} (-1)^{i+1} \binom {N} {i} \mathbb{E}_{I_{}{\| \mathbf{x}_1\|}^{\beta}} \exp \left( - i \alpha \tau I_{}{\| \mathbf{x}_1\|}^{\beta} \right) \notag \\
&=\sum_{i=1}^{N} (-1)^{i+1} \binom {N} {i} \mathbb{E}_{{\| \mathbf{x}_1\|}^{\beta}} \mathcal{L}_{I_{}} \left(i \alpha \tau {\|\mathbf{x}_1\|}^{\beta} \right) \notag \\
&\overset{(b)} {=} \sum_{i=1}^{N}  \frac { (-1)^{i+1} \binom {N} {i}  t_f} {t_f \left(\frac{  2 \tau  } {\beta-2} {}_2 F_1 \left[1,1-\frac{2} {\beta};2-\frac{2} {\beta};-i \alpha \tau\ \right] +1 \right)+(1-t_f) \frac {2 \pi} {\beta} \csc \left(\frac{2 \pi} {\beta} \right) (i \alpha \tau)^{2 \backslash \beta} }.
\end{align}
where (a) follows from $g_{1} \sim \textrm{Gamma}(N,1)$ and Lemma \ref{Gammaupperbound}, (b) follows from the PDF of $\| \mathbf{x}_1\|$ $f_{\| \mathbf{x}_1\|}(r)= 2 \pi t_f \lambda_b r \exp \left \{-\pi t_f \lambda_b r^{2} \right \} $ for $f \in \mathcal{F}_c$ and $\mathcal{L}_{I_{}} \left(i \tau {\|\mathbf{x}_1\|}^{\beta}\right)$ is given in  Appendix \ref{ProofMRTperformance} as $\mathcal{L}_{I_{}} \left(i \tau {\|\mathbf{x}_1\|}^{\beta}\right)=\exp \left(- \pi \lambda_b  \left( \frac{  2 \tau t_f } {\beta-2} {}_2 F_1 \left[1,1-\frac{2} {\beta};2-\frac{2} {\beta};-i \alpha \tau\ \right] +(1-t_f) \frac {2 \pi} {\beta} \csc \left(\frac{2 \pi} {\beta} \right) (i \alpha \tau)^{2 \backslash \beta} \right)  {\|\mathbf{x}_1\|}^2  \right)$. Here, $\alpha={(N!)}^{-\frac {1} {N}}$.

Therefore, we obtain an upper bound of $P_{\textrm{s}}^{u,c,f} \left( t_f \right),  \forall f \in \mathcal{F}_c$. We can obtain the expression of $P_{\textrm{s}}^{u,b}$ similarly and then finish the proof of Theorem \ref{upperboundMRT}.
\end{proof}

The upper bound of the STP is a series of fractional functions of the cache placement $\mathbf{t}$. Comparing with the exact expression of the STP in Theorem \ref{MRTsuccessP}, the upper bound approximates the STP in a simpler manner, and therefore facilitates the analysis and further optimization.

According to Theorem \ref{upperboundMRT}, the upper bound of the ASE is given by
\begin{align}\label{ASEMRTup}
R_{u} \left(\mathcal{F}_c,\mathbf{t} \right)=\lambda_b  P_{\textrm{s}}^u(\mathcal{F}_c,\mathbf{t}) \log_2(1+\tau).
\end{align}

To obtain design insights, we then analyze the upper bound of the ASE in the asymptotic region, i.e., the high user density region. When $\lambda_u \rightarrow \infty$, the discrete random variable $F_{b,1}^r \rightarrow F_b$ in distribution  as shown in Corollary \ref{asymloads}. Therefore, we have the following corollary according to Theorem \ref{upperboundMRT}.
\begin{corollary}\label{highupperboundMRT}
(Asymptotic Upper Bound of ASE) In high user density region, i.e., $\lambda_u \rightarrow \infty$, the asymptotic upper bound of  the ASE is given by

\begin{align}\label{PMRTinfiupperhigh}
R_{u,\infty}(\mathcal{F}_c,\mathbf{t})=\lambda_b \log_2(1+\tau) \bigg( \sum_{f \in \mathcal{F}_c} q_f P_{\textrm{s}}^{u,f,c} (t_f) + \sum_{f \in \mathcal{F} \setminus \mathcal{F}_c}  \frac { q_f  B} {\max \left( F_b,B \right) } P_{\textrm{s}}^{u,b} \bigg),
\end{align}
where $P_{\textrm{s}}^{u,f,c} (t_f)$ and $P_{\textrm{s}}^{u,b}$ are given in (\ref{uppercachefile}) and (\ref{upperbackhaulfile}).
\end{corollary}

We plot Fig. \ref{performanceall}  and  Fig. \ref{Performanceaym} to validate the correctness of the analytic results. Fig. \ref{performanceall} plots the successful transmission probability vs. the number of BS antennas and target SIR. Fig. \ref{performanceall} verifies Theorem \ref{MRTsuccessP} and  Theorem \ref{upperboundMRT}, and demonstrates the tightness of the upper bound. Fig. \ref{performanceall} also shows that the successful transmission probability increases with the number of BS antennas and decreases with the SIR threshold. Moreover, Fig. \ref{performanceall} (a) indicates that when the user density is large, the increase of the BS density increases the successful transmission probability.  Fig. \ref{Performanceaym} plots the area spectrum efficiency vs. user density by showing that when the user density is larger than a certain threshold, i.e., $6 \times 10^{-3} \mathrm{m}^{-2}$ for the single-antenna networks and $4 \times 10^{-3} \mathrm{m}^{-2}$ for the multi-antenna networks, the asymptotic upper bound of the ASE is nearly same as the ASE.

\begin{figure}

\begin{minipage}[t]{0.5\linewidth}
\centering
    \includegraphics[height=2in,width=3in]{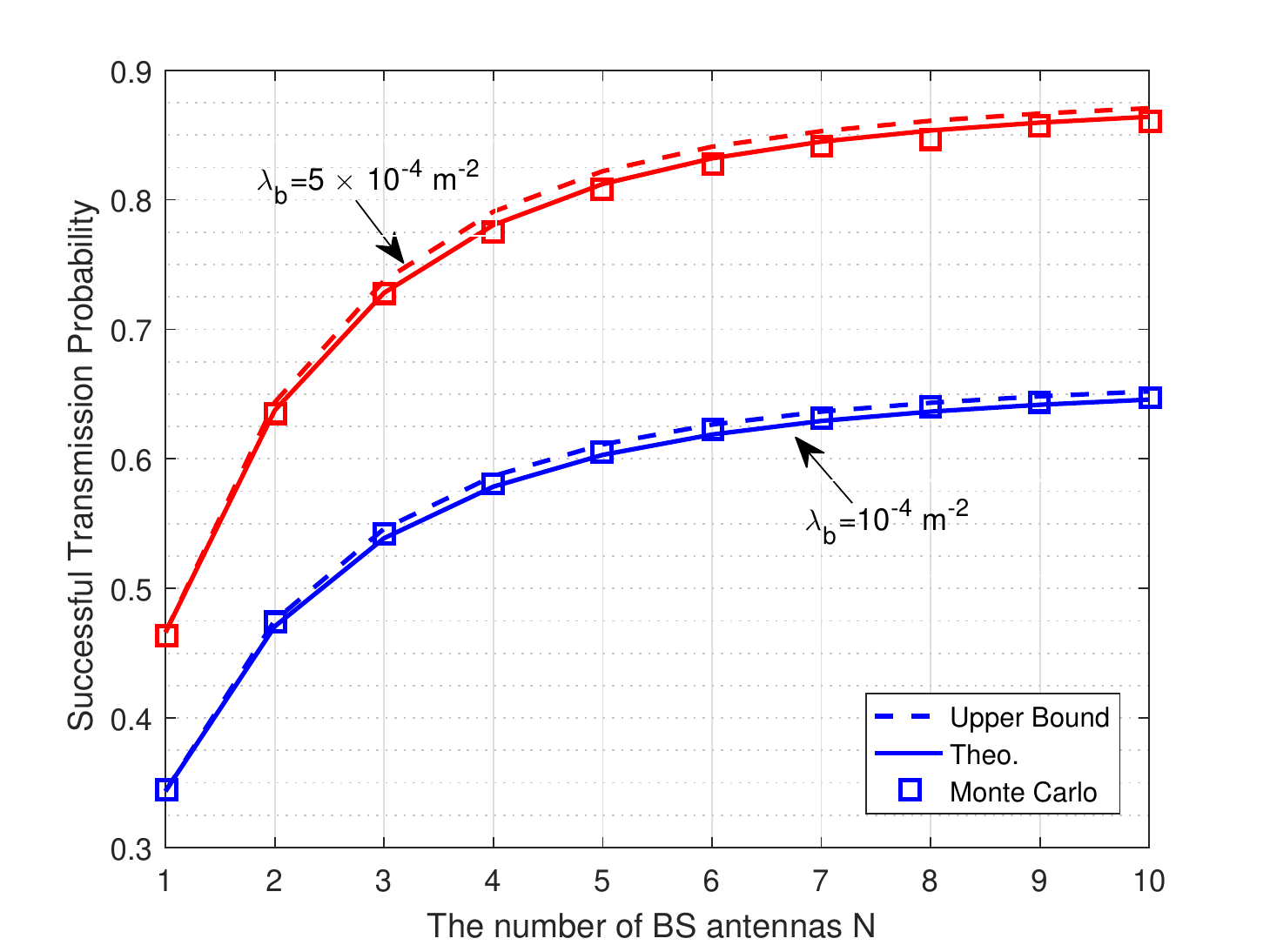}
    \caption*{\footnotesize (a) STP vs. the number of BS antennas $N$ at $\tau=0dB$.}
\end{minipage}
\begin{minipage}[t]{0.5\linewidth}
\centering
    \includegraphics[height=2in,width=3in]{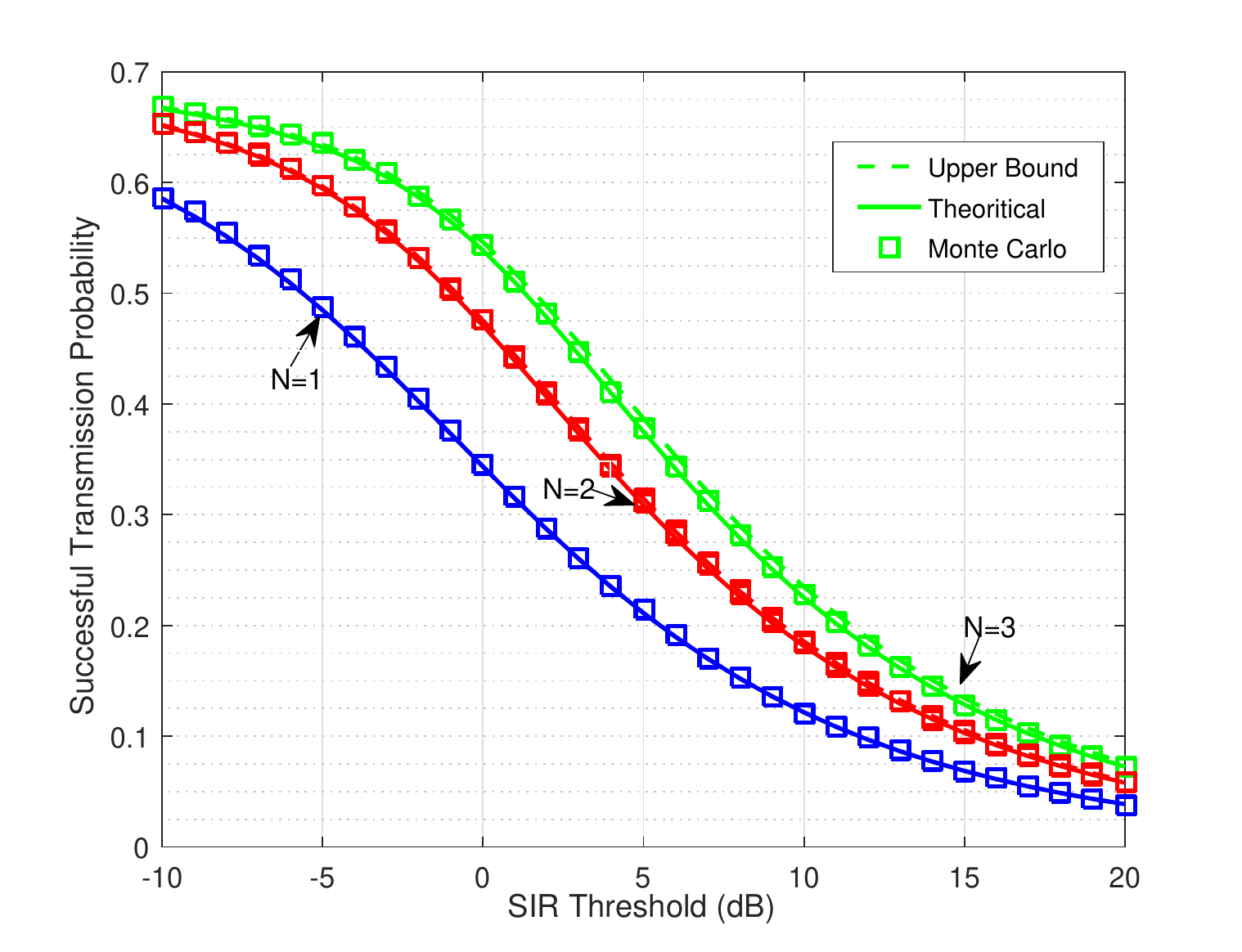}
    \caption*{\footnotesize (b) STP vs.  SIR threshold $\tau$ at $\lambda_b=10^{-4} \mathrm{m}^{-2}$.}
\end{minipage}
\caption{\small STP vs. the number of BS antennas $N$ and SIR threshold $\tau$. $\lambda_u=10^{-3} \mathrm{m}^{-2}$, $\beta=4$, $F=8$, $B=C=2$, $\mathcal{F}_b=\{1,2,3,4\}$, $\mathcal{F}_c=\{5,6,7,8\}$, $\mathbf{t}=(0.8,0.6,0.4,0.2)$, $\gamma=1$. In this paper,  the transmit power is $6.3$W, the noise power in the Monte Carlo simulation is $\sigma_n=-97.5\mathrm{dBm}$ \cite{holma2009lte}, the theoretical results are obtained without consideration of noise. The Monte Carlo results are obtained by averaging over $10^6$ random realizations.}\label{performanceall}
\end{figure}

\begin{figure}
\centering
    \includegraphics[height=2in,width=3in]{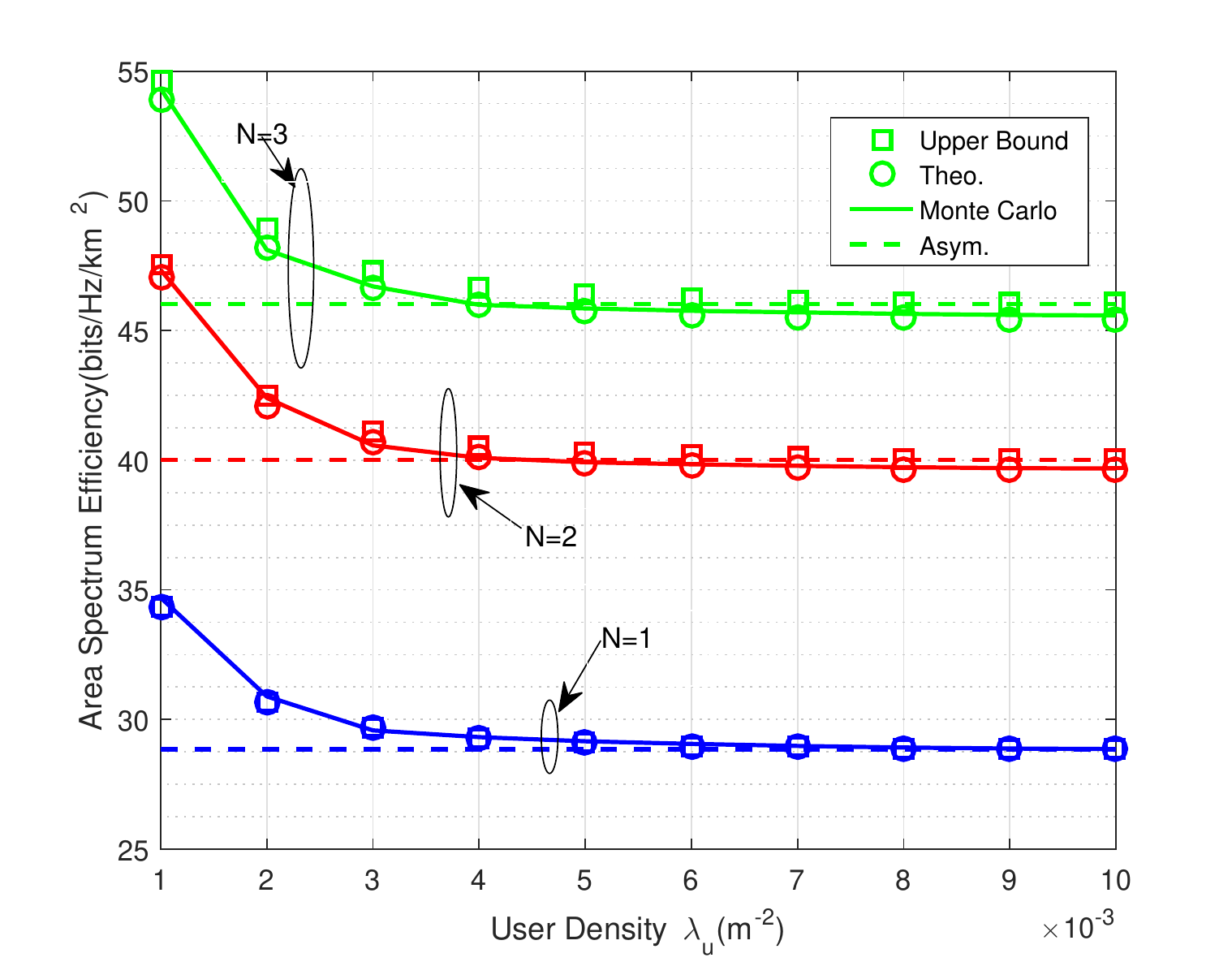}
\caption{The ASE vs. user density $\lambda_u$. $\lambda_b=10^{-4} \mathrm{m}^{-2}$, $\beta=4$, $F=8$, $B=C=2$, $\mathcal{F}_b=\{1,2,3,4\}$, $\mathcal{F}_c=\{5,6,7,8\}$, $\mathbf{t}=(0.8,0.6,0.4,0.2)$, $\gamma=1$, $\tau=0\mathrm{dB}$.} \label{Performanceaym}
\label{Performanceaym}
\end{figure}

\section{ASE Optimization}\label{ASEoptimization}

\subsection{General ASE Optimization}
In this part, we solve the ASE maximization problem, i.e., maximize $R_{}(\mathcal{F}_c,\mathbf{t})$ via optimizing  the file allocation $\mathcal{F}_c$  and the cache placement $\mathbf{t}$. Based on the relationship between $R_{}(\mathcal{F}_c,\mathbf{t})$ and  $P_{\textrm{s}}\left(\mathcal{F}_c , \mathbf{t}\right)$, the ASE optimization problem is formulated as follows.
\begin{problem}\label{problemMRTorigin}
(ASE Optimization)
\begin{align}
R_{}^* \triangleq  &\max_{\mathcal{F}_c,\mathbf{t}} ~~\lambda_b  P_{\textrm{s}}\left(\mathcal{F}_c , \mathbf{t}\right) \log_2(1+\tau)  \\
 &s.t. ~~~(\ref{fileconstraint}),(\ref{T01}),(\ref{Ttotal}).  \notag
\end{align}
where $ P_{\textrm{s}}(\mathcal{F}_c,\mathbf{t})$ is given in (\ref{MRTsuccessall}).
\end{problem}

The above problem is a mixed-integer problem to optimize the discrete parameter $\mathcal{F}_c$ and the continuous parameter $\mathbf{t}$. To solve the complex problem, we first explore optimal properties of
the discrete variable $\mathcal{F}_c$ and then optimize the continuous variable $\mathbf{t}$. To optimally design the file allocation $\mathcal{F}_c$, we first  study the properties of the STPs of the cached file and backhaul file, i.e., $P_{\textrm{s}}^{f,c} \left(t_f\right)$ and $P_{\textrm{s}}^{b}$. Based on the properties of the lower triangular Toeplitz matrix form in (\ref{PsuccessfC}) and (\ref{PsuccessfB}), we then obtain the following lemma.

\begin{lemma}\label{propertiesSTP}
$P_{\textrm{s}}^{f,c} (t_f)$ and $P_{\textrm{s}}^{b}$ have following properties:

1) $P_{\textrm{s}}^{f,c} (t_f)$ is bounded by
\begin{align}\label{STPupperlowerTf}
\frac{t_f} {\mu_A t_f+\nu_A }\leq P_{\textrm{s}}^{f,c} \left( t_f \right) \leq \frac{t_f} {\mu_B t_f+\nu_B},
\end{align}
where $\mu_A$, $\nu_A$, $\mu_B$ and $\nu_B$ are given by
\begin{align}
\mu_A&= 1- \frac {2 \pi} {\beta} \csc \left(\frac{2 \pi} {\beta} \right) \tau^{2 \backslash \beta} +\frac{  2\tau} {\beta-2} {}_2 F_1 \left[1,1-\frac{2} {\beta};2-\frac{2} {\beta};-\tau\ \right] \notag \\
&+\sum_{i=1}^{N-1} \frac { N-i} {N}  \left[\frac{2 \tau^{i}} {i \beta-2} {}_2 F_1 \left[i+1,i-\frac{2} {\beta};i+1-\frac{2} {\beta};-\tau\ \right] - \frac{2 } {\beta} B(\frac{2} {\beta}+1,i-\frac{2} {\beta}) \tau^{2 \backslash \beta} \right], \\
\nu_A&=\frac {2 \pi} {\beta} \csc \left(\frac{2 \pi} {\beta} \right) \tau^{2 \backslash \beta} -\sum_{i=1}^{N-1} \frac { N-i} {N}  \frac{2 } {\beta} B(\frac{2} {\beta}+1,i-\frac{2} {\beta}) \tau^{2 \backslash \beta},  \\
\mu_B&= 1- \frac {2 \pi} {\beta} \csc \left(\frac{2 \pi} {\beta} \right) \tau^{2 \backslash \beta} +\frac{  2\tau} {\beta-2} {}_2 F_1 \left[1,1-\frac{2} {\beta};2-\frac{2} {\beta};-\tau\ \right] \notag \\
&+\sum_{i=1}^{N-1}   \left[\frac{2 \tau^{i}} {i \beta-2} {}_2 F_1 \left[i+1,i-\frac{2} {\beta};i+1-\frac{2} {\beta};-\tau\ \right] - \frac{2 } {\beta} B(\frac{2} {\beta}+1,i-\frac{2} {\beta}) \tau^{2 \backslash \beta}  \right] ,\\
\nu_B&=\frac {2 \pi} {\beta} \csc \left(\frac{2 \pi} {\beta} \right) \tau^{2 \backslash \beta}  -\sum_{i=1}^{N-1}   \frac{2 } {\beta} B(\frac{2} {\beta}+1,i-\frac{2} {\beta}) \tau^{2 \backslash \beta}  ).
\end{align}
Furthermore,  $\nu_A$, $\nu_B$ are positive and $\mu_A$, $\mu_B$ are no larger than $1$.

2) $P_{\textrm{s}}^{f,c} (t_f)$ is an increasing function of $t_f$.

3) $P_{\textrm{s}}^{f,c} (t_f)$ is no larger than $P_{\textrm{s}}^{b}$ and no smaller than $0$, i.e., $0 \leq P_{\textrm{s}}^{f,c} (t_f) \leq P_{\textrm{s}}^{b}$.
\end{lemma}
\emph{Proof:} See Appendix \ref{ProofSTPproperties}.

Property 1 in Lemma \ref{propertiesSTP} shows that the bounds of $P_{\textrm{s}}^{f,c} (t_f)$ are  fractional functions of $t_f$. The expressions of the bounds are similar to the expression of the STP of the cached file in single-antenna networks given in Theorem \ref{STPSISO}.   Based on the properties of $P_{\textrm{s}}^{f,c} (t_f)$ and $P_{\textrm{s}}^{b}$, we  have the following theorem to reveal the properties of the optimal file allocation $\mathcal{F}_c^{*}$.
\begin{theorem}\label{optimalFc}
(Property of the Optimal File Allocation $\mathcal{F}_c^{*}$)
To optimize cached file set $\mathcal{F}_c$  for Problem \ref{problemMRTorigin}, we should allocate at least $C$ files and at most $F-B$ files to store in the cache, that is to say, $F_c^{*}=\{C,C+2,\cdots,F-B\}$.
\end{theorem}

\begin{proof}   We utilize the contradiction to prove Theorem \ref{optimalFc}. More specifically, we consider the cases in which we cache less than $C$ files or more than $F-B$ files, and we then prove that the cases are not optimal in terms of the ASE. For details, please see Appendix \ref{ProofofFcallocation}. \end{proof}

The above theorem interprets that for optimal file allocation, the number of cached files should be larger than the cache size and the number of the backhaul files shoud be larger than the backhaul capability, indicating that we should make full use of the resources.

Based on Theorem \ref{optimalFc}, we can reduce the complexity of search for $\mathcal{F}_c^{*}$. Otherwise, we have to check the cases such that $F_c <C$ or $F_c>F-B$. When $B$ or $C$ is large, Theorem \ref{optimalFc} will largely reduce the complexity. Under given $\mathcal{F}_c$, we optimize the cache placement $\mathbf{t}$. The ASE optimization over $\mathbf{t}$ under given $\mathcal{F}_c$ is formulated as follows
\begin{problem}\label{problemTgivenFc}
(Cache Placement Optimization under Given $\mathcal{F}_c$)
\begin{align}
R^*(\mathcal{F}_c) \triangleq  &\max_{\mathbf{t}} ~~\lambda_b  \sum_{f \in \mathcal{F}_c} q_f  P_{\textrm{s}}^{f,c} \left(t_f\right)  \log_2(1+\tau)  \\
 &s.t. ~~~(\ref{T01}),(\ref{Ttotal}).  \notag
\end{align}
where $P_{\textrm{s}}^{f,c} \left(t_f\right)$ is given in (\ref{PsuccessfC}).

\end{problem}

The optimal solution for Problem \ref{problemTgivenFc}  is denoted as $ \mathbf{t}^* \left( \mathcal{F}_c \right)$. The above problem  is a continuous optimization problem of a differentiable function  over a convex set and we can use the gradient projection method to obtain a local optimal solution.  Under given $\mathcal{F}_c$, we can obtain optimal cache placement $\mathbf{t}^{*} \left(\mathcal{F}_c \right)$ using Algorithm \ref{gradientorigin}.

\begin{algorithm}
\caption{Optimal Solution to Problem \ref{problemTgivenFc}}
\algsetup{linenodelimiter=.}
\begin{algorithmic}[1]\label{gradientorigin}

\STATE \textbf{Initialization}: $n=1$, $n_{\textrm{max}}=10^{4}$ and $t_i(1)=\frac {1} {F_c}$ for all $i \in \mathcal{F}_c$. Constant lower triangular Toeplitz matrix $\frac {\partial  \mathbf{B} } {\partial t_i}$ is given in (\ref{expressionK}).
\REPEAT
\STATE Calculate $ \mathbf{D}_{}^{c,i}$ with $t_i=t_i(n)$ according to (\ref{matrixT}) for all $i \in \mathcal{F}_c$.
 \STATE $\mathbf{B}_i(n)=  \left(t_i(n)+ l_{0}^{c,i} \right)\mathbf{I}- {\tau^{2 \backslash \beta}} \  \mathbf{D}_{}^{c,i} $ for all $i \in \mathcal{F}_c$.
 \STATE $t_i^{'}(n+1)=t_i(n) + s(n) \left \|{\mathbf{B}_i^{-1}(n) } - t_i(n)  \mathbf{B}_i^{-1}(n) \frac {\partial  \mathbf{B} } {\partial t_i}\mathbf{B}_i^{-1}(n) \right \|_1 \lambda_b \log_2(1+\tau) q_i $ for all $i \in \mathcal{F}_c$. \label{matrixinverse}
 \STATE $t_i\left(n+1\right)=\left[ t_i^{'}(n+1)-u^* \right]_0^1$ for all $i \in \mathcal{F}_c$, where $u^*$ satisfying $\sum_{i \in \mathcal{F}_c} \left[ t_i^{'}(n+1)-u^* \right]_0^1=C$ and $[x]_0^1$ denotes $\max \{ \min\{1,x\},0 \}$.
\STATE $n=n+1$.
\UNTIL{Convergence or $n$ is larger than $n_{\textrm{max}}$.}
\end{algorithmic}
\end{algorithm}

From Corollary \ref{STPSISO}, we can easily observe that when $N=1$, i.e., in the single-antenna networks,  the objective function of Problem \ref{problemTgivenFc} is concave and the Slaters condition is satisfied, implying that strong duality holds. In this case, we can obtain a closed-form optimal solution to the convex optimization problem using KKT conditions. After some manipulations, we obtain the optimal cache placement $\mathbf{t}^{*}$ under given $\mathcal{F}_c$ for single-antenna networks.

\begin{corollary}\label{OptimalTgivenFc}
(Optimal Cache Placement of Single-Antenna Networks) When the cached file set $\mathcal{F}_c$ is fixed, the optimal cache placement $\mathbf{t}^*$  of the single-antenna networks is given by
\begin{align}
t_f^*=\left[\frac{1} {\zeta_1 (\tau)} \left(\sqrt{\frac {\lambda_b \log_2(1+\tau) q_f \zeta_2 (\tau)} {u^*}}-\zeta_2 (\tau)\right) \right]_0^1,~f \in {\mathcal{F}_c}, \label{opMRTasymT}
\end{align}
where $[x]_0^1$ denotes $\max \{ \min\{1,x\},0 \}$ and $u^*$ is the optimal dual variable for the constraint  $\sum_{i \in \mathcal{F}_C^*} t_i \leq C$, which satisfies
\begin{align}\label{opMRTasymTdualsolution}
\sum_{f \in {\mathcal{F}_C}} \left[\frac{1} {\zeta_1 (\tau)} \left(\sqrt{\frac {\lambda_b \log_2(1+\tau) q_f  \zeta_2 (\tau) } {u^*}}-\zeta_2 (\tau) \right) \right]_0^1=C.
\end{align}
Here $\zeta_1 (\tau)$ and $\zeta_2 (\tau)$ are given in (\ref{expressionzeta1}) and (\ref{expressionzeta2}), respectively.
\end{corollary}

Finally, combining the analysis of the properties of $\mathcal{F}_c^*$ in Theorem \ref{optimalFc} and the optimization of $\mathbf{t}$ under given $\mathcal{F}_c$, we can obtain $\left(\mathcal{F}_c^*,\mathbf{t}^* \right)$ to Problem \ref{problemMRTorigin}. The process of solving  Problem \ref{problemMRTorigin} is summarized in Algorithm \ref{Optimalexactall}.

\begin{algorithm}
\caption{Optimal Solution to Problem \ref{problemMRTorigin}}
\algsetup{linenodelimiter=.}
\begin{algorithmic}[1]\label{Optimalexactall}

\STATE \textbf{Initialization}: $R_{}^*=\infty$.
\FOR{$F_c=C:F-B$ }    \label{searchfile}
 \STATE Choose $\mathcal{F}_c \in \left \{\mathcal{X} \subseteq \mathcal{F} : |\mathcal{X} |=F_c \right \}$. \label{searchfilein}
 \STATE Obtain the optimal solution $\mathbf{t}^{*} \left(\mathcal{F}_c \right)$  to Problem \ref{problemTgivenFc} using Algorithm 1 (when $N>1$) or Corollary \ref{OptimalTgivenFc} (when $N=1$).
 \IF {$R_{}^*<R_{}\left(\mathcal{F}_c , \mathbf{t}^{*} \left(\mathcal{F}_c \right)\right)$}
 \STATE $R_{}^*=R_{}\left(\mathcal{F}_c , \mathbf{t}^{*} \left(\mathcal{F}_c \right)\right)$,  $\left(\mathcal{F}_c^{*},\mathbf{t}^{*} \right)=\left(\mathcal{F}_c , \mathbf{t}^{*} \left(\mathcal{F}_c \right)\right)$.
  \ENDIF
\ENDFOR
\end{algorithmic}
\end{algorithm}

 Algorithm \ref{Optimalexactall} includes two layers. In the outer layer, we search $\mathcal{F}_c^*$ by checking all the possible  $\sum_{i=C}^{F-B} \binom{F}{i}$ choices, and in the inner layer, we utilize Algorithm \ref{gradientorigin} or Corollary \ref{OptimalTgivenFc} to obtain $\mathbf{t}^* \left( \mathcal{F}_c \right)$ under given $\mathcal{F}_c$.
We refer to the optimal solution based on  Algorithm \ref{Optimalexactall} as \emph{Exact Opt.}.

\subsection{Asymptotic ASE Optimization based on the Upper Bound in (\ref{PMRTinfiupperhigh})}
In Algorithm \ref{Optimalexactall}, we need to consider $\sum_{i=C}^{F-B} \binom{F}{i}$ choices and the complexity is $\Theta(F^F)$. When $F$ is very large,  Algorithm \ref{Optimalexactall} is not acceptable due to high complexity. Furthermore, the expression of $P_{\textrm{s}}^{f,c} \left(t_f\right)$ is complex and we need to calculate the matrix inverse to obtain the derivative of $t_f$. Thus Algorithm \ref{Optimalexactall} requires high complexity if the number of antennas is large. Note that the upper bound of the STP closely approximates the STP, as illustrated in Fig. \ref{performanceall}, therefore we utilize the upper bound of the STP, i.e., $P_{\textrm{s}}^{u,f,c} \left(t_f\right)$ to approximate $P_{\textrm{s}}^{f,c} \left(t_f\right)$.  To facilitate the optimization, we formulate the problem of optimizing the asymptotic upper bound of the ASE to provide insightful guidelines for the parameter design in the high user density region.  Based on the Corollary \ref{highupperboundMRT}, the asymptotic optimization problem is formulated as follows.
\begin{problem}\label{ProblemMRTasymT}
(Asymptotic ASE Optimization)
\begin{align}
 R_{u,\infty}^* \triangleq \max_{\mathcal{F}_c, \mathbf{t}}& ~~ R_{u,\infty} \left(\mathcal{F}_c,\mathbf{t} \right) \notag \\
 s.t & ~~~(\ref{fileconstraint}),(\ref{T01}),(\ref{Ttotal}).
\end{align}
\end{problem}

The above problem is a mixed-integer problem.  By carefully investigating the characteristic of $R_{u,\infty} \left(\mathcal{F}_c,\mathbf{t} \right)$, which is a series of fractional functions of $\mathbf{t}$ in (\ref{PMRTinfiupperhigh}), we obtain the following lemma to reveal the properties of $R_{u,\infty} \left(\mathcal{F}_c,\mathbf{t} \right)$.
\begin{lemma}\label{upperboundproperty}
Under a given $\mathcal{F}_c$, $R_{u,\infty} \left(\mathcal{F}_c,\mathbf{t}\right)$ is an increasing function of $t_f$ for any $f \in \mathcal{F}_c$.
\end{lemma}

\emph{Proof:} See Appendix \ref{Proofpropertyupper}.

Based on the properties of the asymptotic upper bound  $R_{u,\infty}(\mathcal{F}_c,\mathbf{t})$,  we then analyze the properties of the optimal file allocation  $\mathcal{F}_c^{*}$ and obtain $\mathcal{F}_c^{*}$ as a unique solution.
\begin{theorem}\label{aymoptFc}
(Asymptotic Optimal File Allocaiton)
In high user density region, i.e., $\lambda_u \rightarrow \infty$, the optimal cached file  $\mathcal{F}_C^*$  is given by $\mathcal{F}_C^*=\{B+1,B+2,\cdots,F \}$.
\end{theorem}
\begin{proof}
 We first prove that the number of optimal cached files is $F-B$ and then prove that the optimal $F-B$ cached files are the least $F-B$ popular files. For details, please see Appendix \ref{ProofofasymTFc}.
\end{proof}

Theorem \ref{aymoptFc} interprets that we should transmit $B$ most popular files via the backhaul and cache the remaining files when the user density is very large.  Compared to the process  of checking $\sum_{i=C}^{F-B} \binom{F}{i}$ choices in Algorithm \ref{Optimalexactall}, we get a unique optimal solution of file allocation and thus largely reduce the complexity. When $\mathcal{F}_c^{*}$ is given, we  only need to optimize the continuous variable $\mathbf{t}$. We then use the gradient projection method get the local optimal solution. The algorithm is summarized in Algorithm \ref{aympalgorithm}.

\begin{algorithm}
\caption{Optimal Solution to Problem \ref{ProblemMRTasymT}} \label{aympalgorithm}
\algsetup{linenodelimiter=.}
\begin{algorithmic}[1]

\STATE \textbf{Initialization}: $n=1$, $n_{\textrm{max}}=10^{4}$, $\mathcal{F}_c^{*}= \left \{B+1,B+2,\cdots,F\right \}$, $t_i(1)=\frac {1} {F_c}$ for all $i \in \mathcal{F}_c^{*}$, $\theta_A \left(j \right)=1+\frac{  2\tau} {\beta-2} {}_2 F_1 \left[1,1-\frac{2} {\beta};2-\frac{2} {\beta};-j \alpha \tau\ \right] -\frac {2 \pi} {\beta} \csc \left(\frac{2 \pi} {\beta} \right) (j \alpha \tau)^{2 \backslash \beta}$,
$\theta_C \left(j \right)=\frac {2 \pi} {\beta} \csc \left(\frac{2 \pi} {\beta} \right) (j \alpha \tau)^{2 \backslash \beta}$ for all $j \in \{1,2,\cdots, N\}$.
\REPEAT
 \STATE $t_i^{'}(n+1)=t_i(n) + s(n) \left( \sum_{j=1}^{N}  \frac { (-1)^{j+1} \binom {N} {j} \theta_C \left(j \right) } {(\theta_A \left(j \right) t_i(n) +\theta_C \left(j \right))^2} \right)\lambda_b \log_2(1+\tau) q_i $ for all $i \in \mathcal{F}_c$.
 \STATE $t_i\left(n+1\right)=\left[ t_i^{'}(n+1)-u^* \right]_0^1$ for all $i \in \mathcal{F}_c$, where $u^*$ satisfying $\sum_{i \in \mathcal{F}_c} \left[ t_i^{'}(n+1)-u^* \right]_0^1=C$ and $[x]_0^1$ denotes $\max \{ \min\{1,x\},0 \}$.
\STATE $n=n+1$.
\UNTIL{Convergence or $n$ is larger than $n_{\textrm{max}}$.}
\end{algorithmic}
\end{algorithm}

\begin{figure}
\centering
    \includegraphics[height=2in,width=3in]{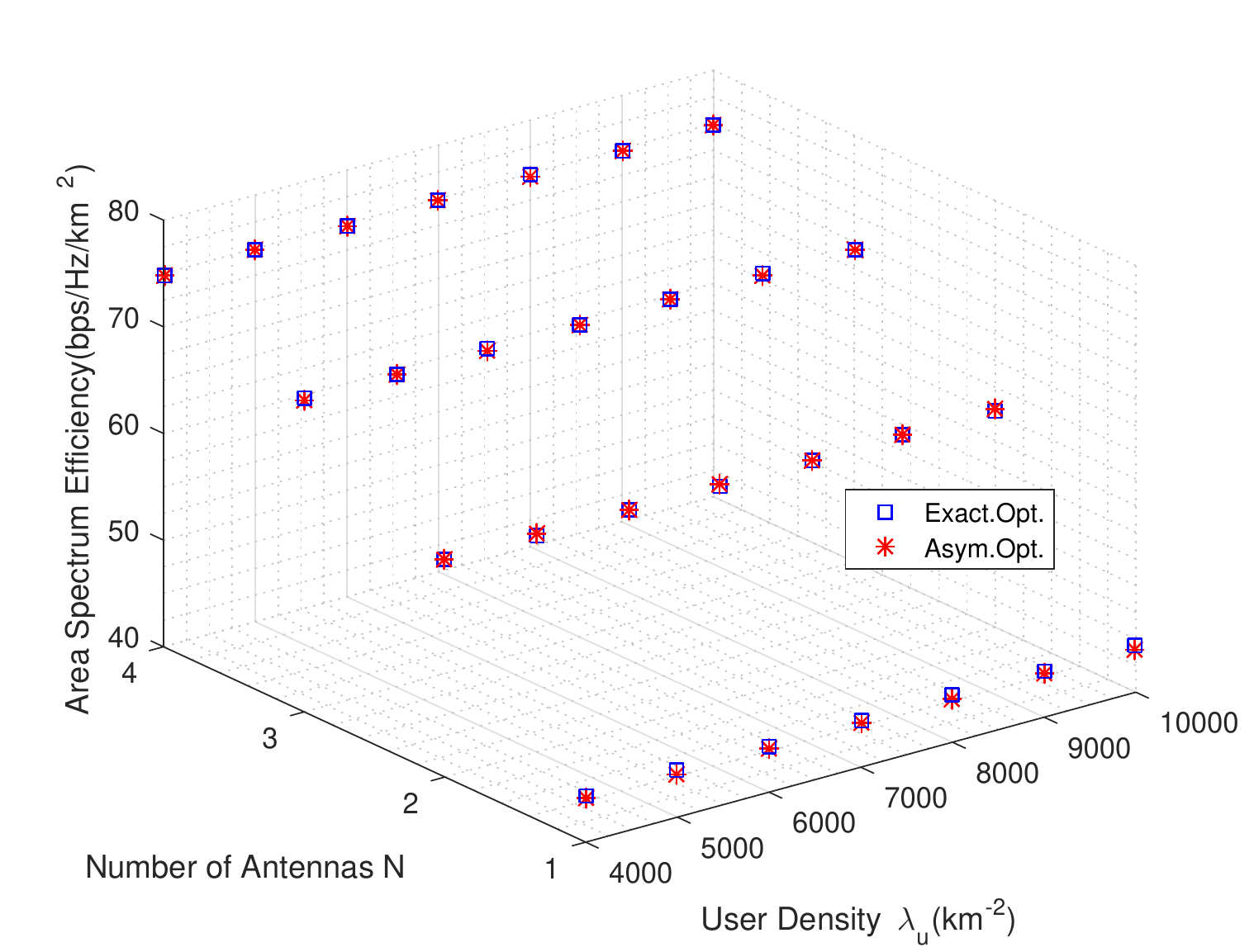}
\caption{Comparison between \emph{Exact Opt.} and \emph{Asym. Opt.}. $\lambda_b=10^{-4} \mathrm{m}^{-2}$, $\beta=4$, $F=6$, $B=C=2$, $\gamma=0.6$, $\tau=0\mathrm{dB}$.}
\label{CompareOptimalresult}
\end{figure}

Algorithm \ref{aympalgorithm} is guaranteed to converge because the gradient projection method converges to a local optimal point for solving a problem whose feasible set is convex. In Algorithm \ref{aympalgorithm}, the calculation of the matrix inverse (Step \ref{matrixinverse} in Algorithm \ref{gradientorigin}) and the search for  $\mathcal{F}_c^{*}$ (Step \ref{searchfilein} in Algorithm \ref{Optimalexactall} ) are avoided. Therefore, Algorithm \ref{aympalgorithm} achieves lower complexity comparing to Algorithm \ref{Optimalexactall}. We refer to the caching scheme based on Algorithm \ref{aympalgorithm} as \emph{Asym. Opt.}.

Now we utilize simulations to compare the proposed \emph{Exact Opt.} (the optimal solution obtained by Algorithm \ref{Optimalexactall}) and \emph{Asym. Opt.} (the asymptotic optimal solution obtained by Algorithm \ref{aympalgorithm}). From Fig. \ref{CompareOptimalresult}, we can see that the performance of  \emph{Asym. Opt.} is very close to that of \emph{Exact Opt.}, even when the user density is low. Therefore, Algorithm \ref{aympalgorithm} with low complexity  is applicable and effective for parameter design in general region.

\section{Numerical Results}\label{numeric results}
In this section, we compare the proposed asymptotic optimal caching scheme given by Algorithm \ref{aympalgorithm} with three caching schemes, i.e., the MPC (most popular caching) scheme \cite{Debbah15cache}, the UC (uniform caching) scheme \cite{15UCcaching} and the IID (identical independent distributed caching) scheme \cite{Nagananda15IIDcache}. In the MPC scheme, the BSs cache or use backhaul to deliver the most popular $B+C$ files. In the IID scheme, the BSs select  $B+C$ files to cache or transmit via the backhaul in an i.i.d manner with probability $q_i$ for file $i$.  In the UC scheme, the BSs select $B+C$ files according to the uniform distribution to cache or deliver via the backhaul. Note that in simulations, we consider the noise. Unless otherwise stated, our simulation environment parameters are as follows:  $P=6.3\mathrm{W}$, $\sigma_n=-97.5\mathrm{dBm}$, $\lambda_b=10^{-4} \mathrm{m}^{-2}$, $\lambda_u=5 \times 10^{-3} \mathrm{m}^{-2}$, $\beta=4$,  $F=500$, $\tau=0\mathrm{dB}$.

Fig. \ref{optimization1} and Fig. \ref{optimization2} illustrate the area spectrum efficiency vs. different parameters. We observe that the proposed asymptotic optimal scheme outperforms all previous caching schemes. In addition, out of the previous caching schemes, the MPC scheme achieves the best performance and the UC scheme achieves the worst performance.

Fig. \ref{optimization1} (a) plots the ASE vs. the number of BS antennas. We can see that the ASE of all the schemes increases with the number of BS antennas. This is because the increase of the number of BS antennas leads to larger spatial diversity and thus achieves better performance. It is shown that the increase of the number of BS antennas leads to an increasing gap between the proposed asymptotic optimal caching scheme and previous caching schemes.  This is because the better performance of a larger number of BS antennas leads to a larger gain when we exploit the file diversity. Furthermore, for asymptotic optimal caching scheme,  the less popular files are more likely to be stored when the number of BS antennas is large. Fig. \ref{optimization1} (b) plots the ASE vs. the Zipf parameter $\gamma$. We can see that the ASE of the proposed asymptotic optimal caching scheme, the MPC scheme and the IID scheme increases with the increase of the Zipf parameter $\gamma$. This is because when $\gamma$ increases, the probability that the popular files are requested increases, and hence, the users are more likely to require the popular files from the nearby BSs who cache the files or obtain the files via backhaul. The change of the Zipf parameter $\gamma$ have no influence to the ASE of the UC scheme. This is because in the UC scheme, all the files are stored/fetched with the same probability, and the change of the file popularity by altering $\gamma$ has no influence to the ASE.

\begin{figure}
\begin{minipage}[t]{0.5\linewidth}
\centering
    \includegraphics[height=2in,width=3in]{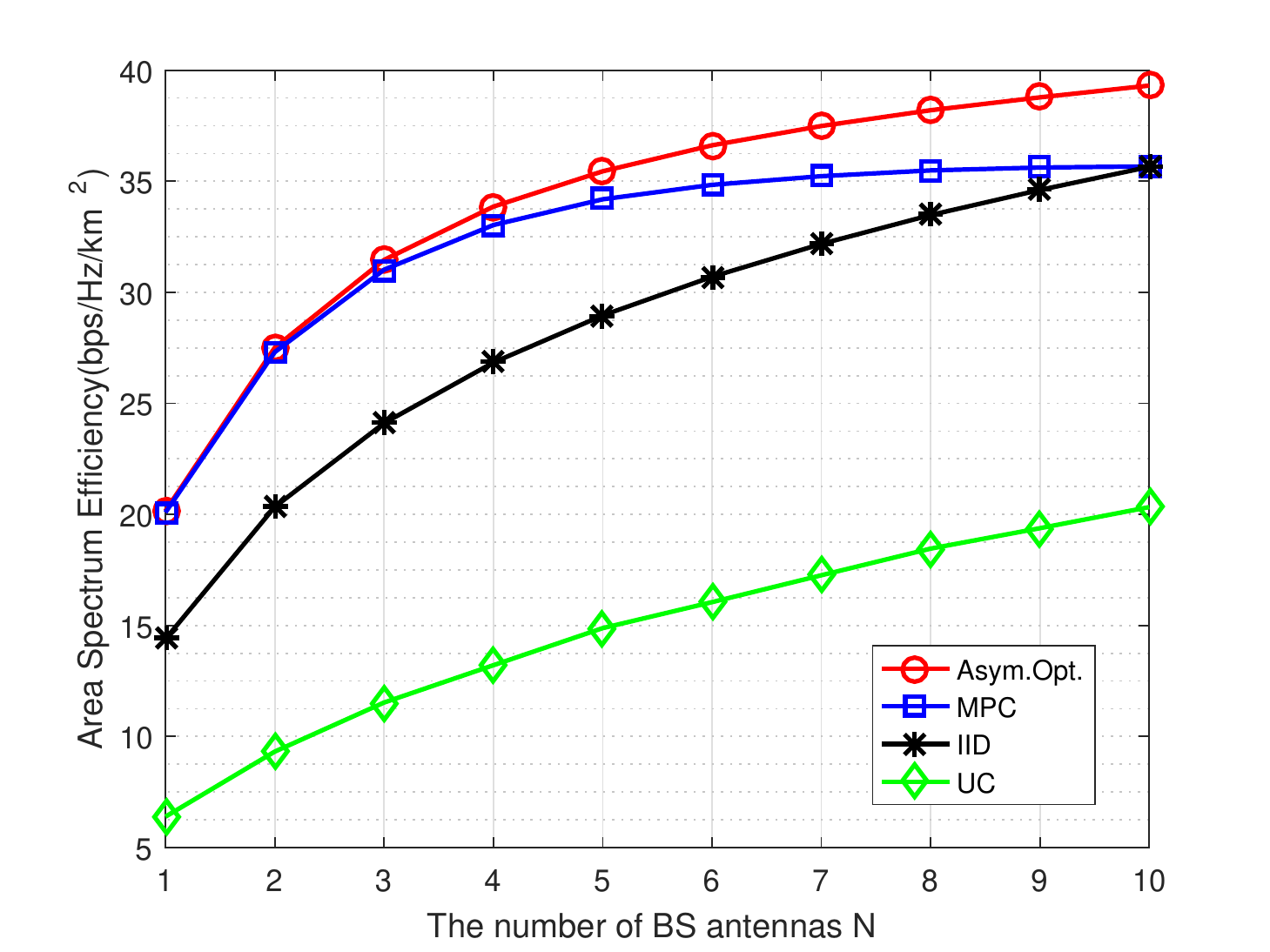}
\caption*{(a) Number of BS antennas at $C=30$, $B=20$, $\gamma=0.6$.}
\end{minipage}
\begin{minipage}[t]{0.5\linewidth}
\centering
     \includegraphics[height=2in,width=3in]{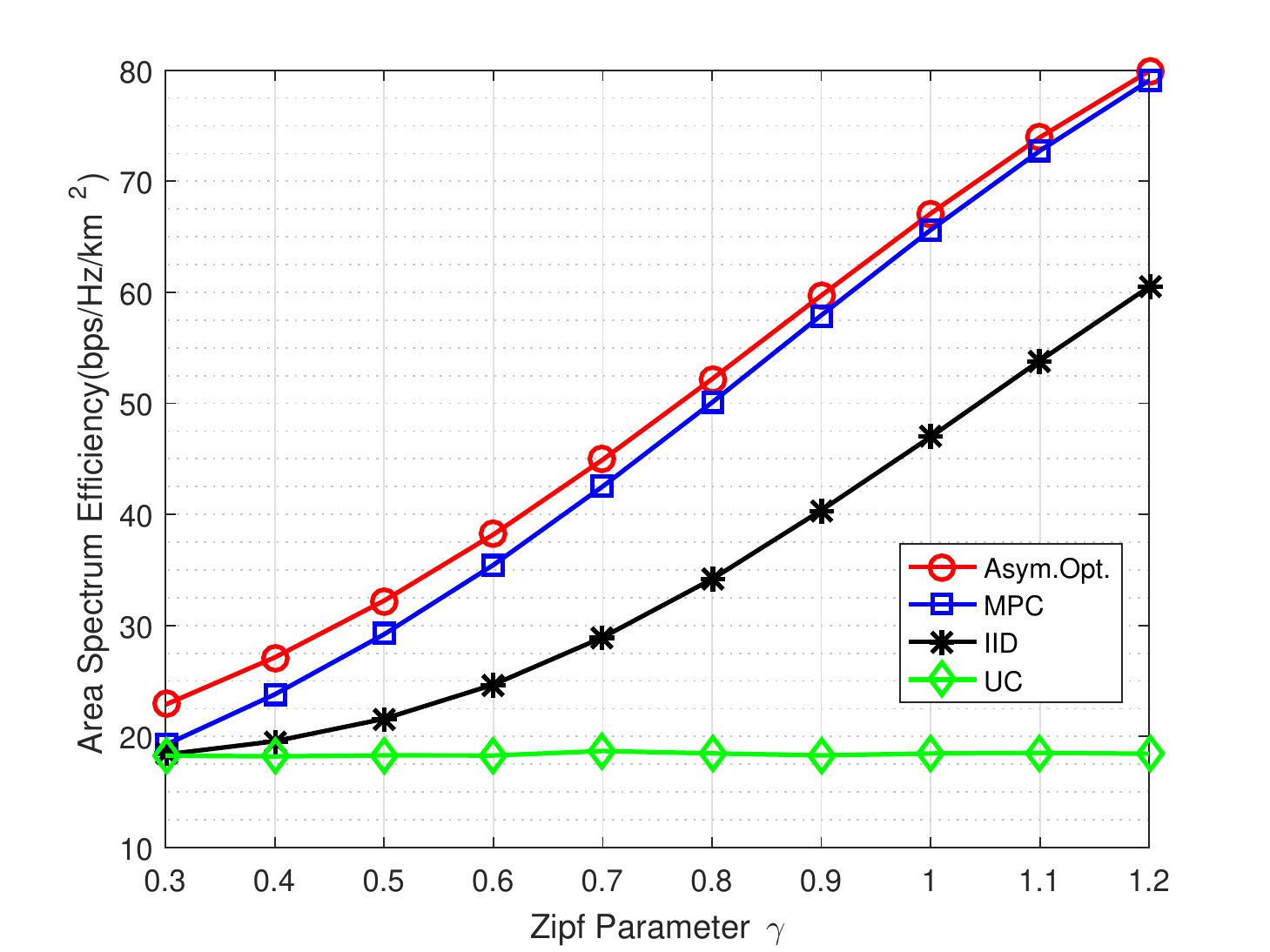}
\caption*{(b) Zipf parameter at $N=8$, $C=30$, $B=20$. }
\end{minipage}
\caption{ASE vs. the number of BS antennas $N$ and Zipf parameter $\gamma$.} \label{optimization1}
\end{figure}

\begin{figure}
\begin{minipage}[t]{0.5\linewidth}
\centering
    \includegraphics[height=2in,width=3in]{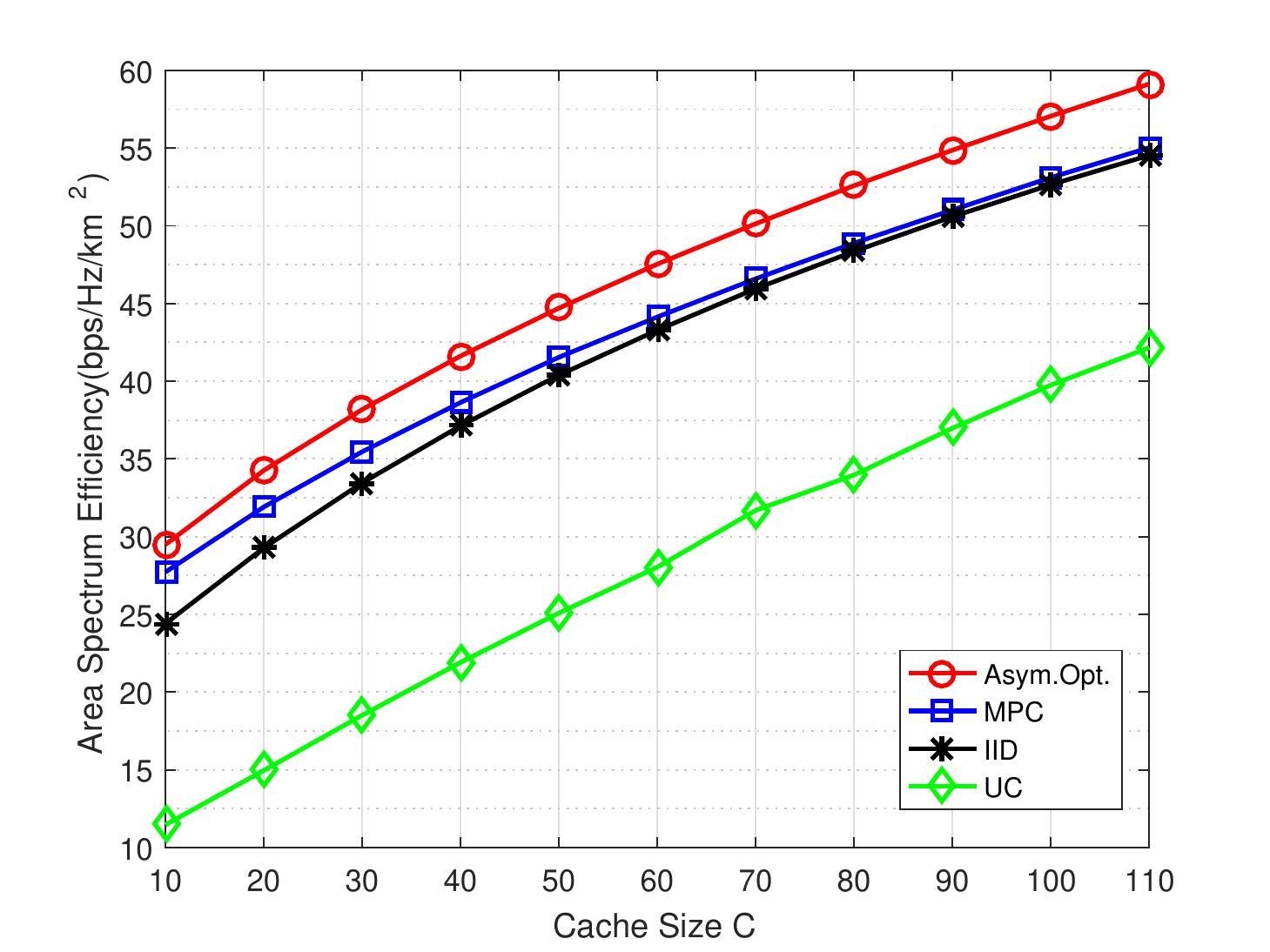}
\caption*{\footnotesize (a) Cache size at $N=8$, $B=20$, $\gamma=0.6$.}
\end{minipage}
\begin{minipage}[t]{0.5\linewidth}
\centering
     \includegraphics[height=2in,width=3in]{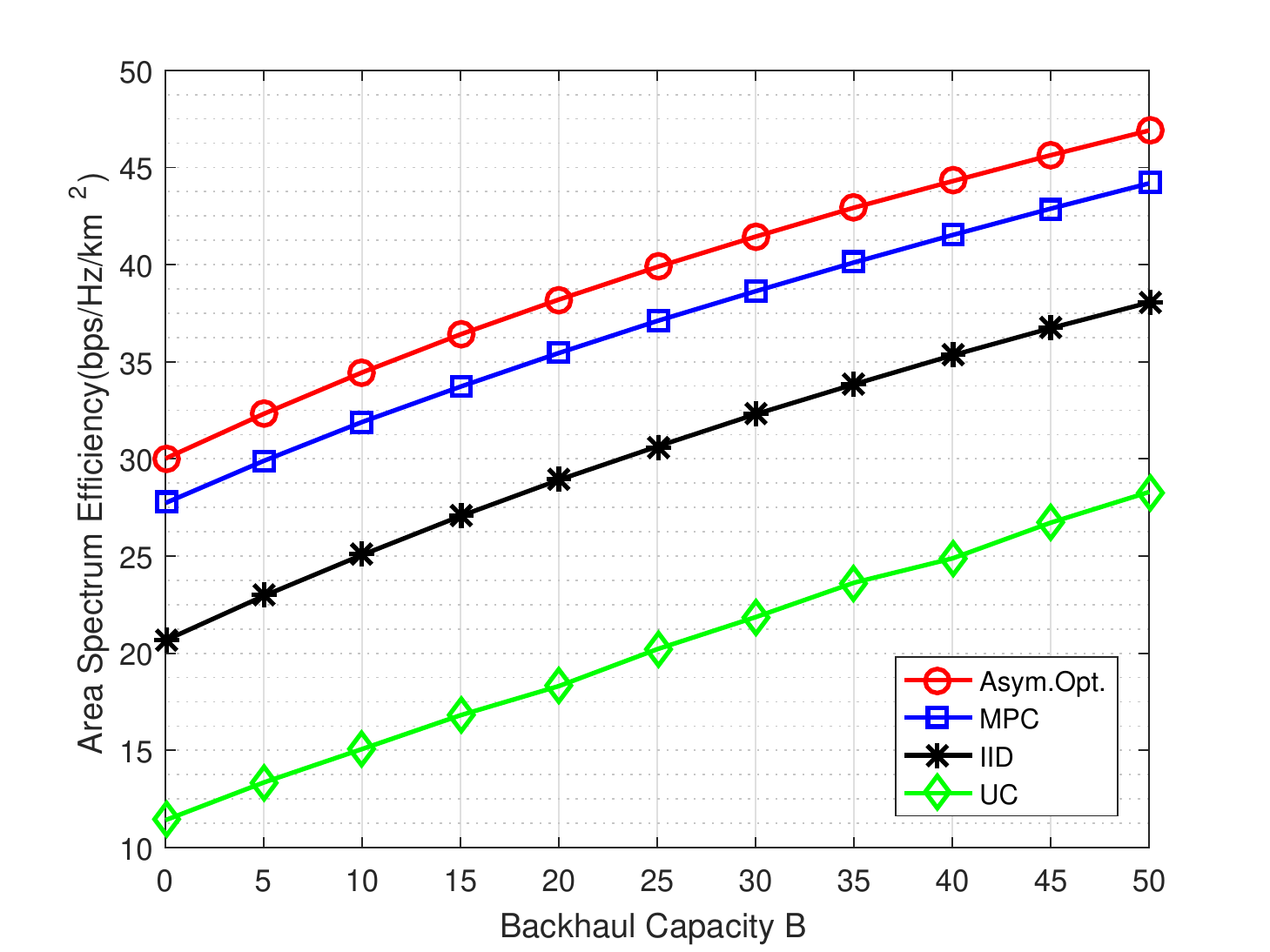}
\caption*{(b) Backhaul capability at $N=8$, $C=30$, $\gamma=0.6$.}
\end{minipage}
\caption{ASE vs. cache size $C$ and backhaul capability $B$. } \label{optimization2}
\end{figure}

Fig. \ref{optimization2} plots the ASE vs. the cache size $C$ or the backhaul capability $B$. From  Fig. \ref{optimization2}, we can see that the ASE of all the schemes increases with the cache size and the backhaul capability because the probability that a randomly requested file is cached at or delivered via the backhaul increases.  The increase of the cache size leads to the increase of the gap between the proposed asymptotic optimal caching scheme and the MPC scheme. This is because when the cache size is small, the ASEs of the proposed caching scheme and the MPC scheme mainly come from the spectrum efficiency of the backhaul files, which is same for the proposed caching scheme and the MPC scheme. The increase of the cache size can bring larger gains of caching diversity. However, the increase of the backhaul capability has no influence to the gap between the proposed caching scheme and the previous caching schemes.

\section{Conclusion}\label{conclusion}
In this paper, we consider the analysis and optimization of random caching in backhaul-limited multi-antenna networks. We propose a file allocation and cache placement design to effectively improve the network performance.  We first derive an exact expression and an upper bound of the successful transmission probability, using tools from stochastic geometry. Then, we consider the area spectrum efficiency maximization problem with continuous variables and integer variables.   We obtain a local optimal solution with reduced complexity by exploring optimal properties, and we also solve an asymptotic optimization problem in high user density region, utilizing the upper bound as the objective function. Finally, we show that the proposed asymptotic optimal caching scheme achieves better performance compared with the existing caching schemes, and the gains are larger when the number of the antennas is larger and/or the Zipf parameter is smaller.

\begin{appendix}
\subsection{Proof of Lemma \ref{fbrspmf} }\label{Proofloadspmf}

According to the thinning theory of the PPP, the density of the users that require file $f$ is $q_f \lambda_u$. We define $P_i$ as the probability that file $i$ is requiring by the BS. Note that $P_i$ is equivalent to the probability that the number of the users requiring file $i$ is not zero, which is $P_i=1-{\left(1+ \frac {q_i  \lambda_u} {3.5  \lambda_b} \right)}^{-4.5}$  according to \cite{Jeffrey13Loads}.  When file $f$ is requested by a certain BS, to calculate the pmf of $F_{b,1}^r$, we need to consider the rest $F_b-1$ files out of $\{\mathcal{F}_b \setminus f \}$ because file $f$ is always requested by the BS. We define the file from the set $\{\mathcal{F}_b \setminus f \}$ as  \emph{rest backhaul file}  and the file from the set $\{\mathcal{F}_{b,1}^r \setminus f \}$ as  \emph{rest backhaul request file}.  To calculate the probability that the number of the rest backhaul request files is $k$, i.e., $\mathbb{P}_f^{\mathcal{F}_b}\left(F_{b,1}^r-1=k \right)$, we combine all the probabilities of the cases when  any given $k-1$ rest backhaul files are required and the $F_b-k$ rest backhaul files are not required. Therefore, $\mathbb{P}_f^{\mathcal{F}_b} \left(F_{b,1}^r=k \right)$ is given by
\vspace{0.1cm}
\begin{align}
\mathbb{P}_f^{\mathcal{F}_b} \left(F_{b,1}^r-1=k \right) = \sum_{\mathcal{Y} \in \left \{\mathcal{X} \subseteq \left \{ \mathcal{F}_b \setminus f \right \} : |\mathcal{X} |=k \right \} }
\prod_{i \in \mathcal{X} } P_i \prod_{i \in \mathcal{B} \setminus \mathcal{X} } \left(1- P_i \right), k=\{0,1,\cdots,F_b-1 \}.
\end{align}
Therefore we finish the proof of \emph{Lemma} \ref{fbrspmf}.

\vspace{-0.4em}

\subsection{Proof of Theorem \ref{MRTsuccessP} }\label{ProofMRTperformance}
The STP of $f \in \mathcal{F}_c$ is given by
\begin{align}
P_{\textrm{s}}^{f,c}(t_f) &=\mathbb{P} \left( \textrm{SIR}_{}^f  >\tau \right) \notag \\
&=\mathbb{P}\left(g_{\textrm{1}}> \tau I_{}{\| \mathbf{x}_1\|}^{\beta}  \right) \notag \\
&\overset{(a)} {=}\mathbb{E}_{\mathbf{x}_1} \left[ \sum_{m=0}^{N-1} \frac {(-1)^m \tau^m {\| \mathbf{x}_1\|}^{m \beta}} {m!}    \mathcal{L}_{I_{}}^{(m)} \left(\tau {\| \mathbf{x}_1\|}^{\beta} \right) \right ],
\end{align}
where (a) follows from $g_{1} \sim \textrm{Gamma}(N,1)$, $I \triangleq \sum_{ i \in \left \{ \Phi_b \setminus 1 \right \} }    \| \mathbf{x}_i \|^{-\beta} g_i $, $\mathcal{L}_{I} (s)
\triangleq \mathbb{E}_{I} \left[ \exp (-Is) \right]$ is the Laplace transform of $I$, and $\mathcal{L}_{I}^{(m)} (s)$ is the $m$th derivative of $\mathcal{L}_{I} (s)$.

Denote $s=\tau {\| \mathbf{x}_1\|}^{ \beta}$ and $y_m^{}=\frac {(-1)^m s^m} {m!} \mathcal{L}_{I_{}}^{(m)} (s)$, we have $P_{\textrm{s}}^{f,c}(t_f) = \mathbb{E}_{\mathbf{x}_1}  \left[\sum_{m=0}^{N-1} y_m \right]$. To distinguish the the interference, we define $I_{}^f=\sum_{i \in \Phi_b^f \setminus B_{f,0}}   \| \mathbf{x}_1 \|^{-\beta} g_{i}$ and $I_{}^{-f}=\sum_{i \in \Phi_b^{-f} \setminus B_{f,0}}   \| \mathbf{x}_1 \|^{-\beta} g_{i}$. We then have $\mathcal{L}_{I_{}} (\tau {\|\mathbf{x}_1\|}^{\beta}))= \mathcal{L}_{I_{}^f} (\tau {\|\mathbf{x}_1\|}^{\beta})) \mathcal{L}_{I_{}^{-f}} (\tau {\|\mathbf{x}_1\|}^{\beta}))$.

Therefore, we can derive the expression of $\mathcal{L}_{I_{}^f} \left(\tau {\|\mathbf{x}_1\|}^{\beta})\right)$ as follows
\begin{align}\label{MRTfLaplace}
\mathcal{L}_{I_{}^f} \left(\tau {\|\mathbf{x}_1\|}^{\beta})\right)&=\mathbb{E} \left[ \exp \left\{ -\tau {\|\mathbf{x}_1\|}^{\beta} \sum_{ i \in \left \{ \Phi_b^f \backslash 1 \right \} }  \| \mathbf{x}_i \|^{-\beta} g_{i} \right\} \right] \notag \\
&\overset{(a)}{=}\prod_{i \in \left \{ \Phi_b^f \backslash 1 \right \} } \mathbb{E} \left[ \frac {1} {1+\tau {\|\mathbf{x}_1\|}^{\beta} \| \mathbf{x}_i \|^{-\beta} }  \right] \notag \\
&\overset{(b)} {=}\exp\left(-2 \pi t_f \lambda_b \int_{\|\mathbf{x}_1 \|}^{\infty}  \left(1-\frac {1} {1+\tau {\|\mathbf{x}_1\|}^{\beta} r^{-\beta}} \right) r \mathrm{d} r \right)\notag \\
&= \exp \left(- \pi t_f \lambda_b \frac{ 2 \tau} {\beta-2} {}_2 F_1 \left[1,1-\frac{2} {\beta};2-\frac{2} {\beta};-\tau\ \right] {\|\mathbf{x}_1\|}^2  \right),
\end{align}
where (a) follows from $g_{i}  \sim  \exp(1)$ due to the random beamforming effect, (b) follows from the probability generating functional (PGFL) of the PPP \cite{haenggi2012stochastic}.
Similarly, we have
\begin{align}\label{MRTnfLaplace}
\mathcal{L}_{I_{}^{-f}} \left(\tau {\|\mathbf{x}_1\|}^{\beta} \right)&=\prod_{i \in \Phi_b^{-f} } \mathbb{E} \left[ \exp \left\{ -\tau {\|\mathbf{x}_1\|}^{\beta} \| \mathbf{x}_i \|^{-\beta} g_{i}\right\} \right] \notag \\
&=\exp\left(-2 \pi (1-t_f) \lambda_b \int_{0}^{\infty}  \left(1-\frac {1} {1+\tau {\|\mathbf{x}_1\|}^{\beta} r^{-\beta}} \right) r \mathrm{d} r \right)\notag \\
&=\exp \left( - \pi (1-t_f) \lambda_b \frac{2 \pi} {\beta} \csc \left(\frac{2 \pi} {\beta} \right) \tau^{2 \backslash \beta} {\|\mathbf{x}_1\|}^2  \right).
\end{align}

Therefore, the expression of $\mathcal{L}_{I_{}} (\tau {\|\mathbf{x}_1\|}^{\beta}))$ is given by
\begin{align}
&\mathcal{L}_{I_{}} (\tau {\|\mathbf{x}_1\|}^{\beta}))=\exp \left(- \pi \lambda_b  \left( t_f \frac{  2\tau} {\beta-2} {}_2 F_1 \left[1,1-\frac{2} {\beta};2-\frac{2} {\beta};-\tau\ \right] +(1-t_f) \frac {2 \pi} {\beta} \csc \left(\frac{2 \pi} {\beta} \right) \tau^{2 \backslash \beta} \right)  {\|\mathbf{x}_1\|}^2  \right).  \notag
\end{align}

To further calculate the STP, we need to calculate the $n$th derivative of Laplace transform $\mathcal{L}_{I_{}}^{(n)} (s)$. After some calculations, we can derive the following recursive relationship
\begin{align}\label{MRTrecursive}
\mathcal{L}_{I_{}}^{(n)} (s) &= \sum_{i=0}^{n-1} \binom {n-1} {i} (-1)^{n-i} (n-i)!\left[\pi t_f \lambda_b \int_{{\|\mathbf{x}_1 \|}^2}^{\infty} \frac {\left( v^{-\frac{\beta} {2} }\right)^{n-i} \mathrm{d} v } {\left(1+s v^{-\frac{\beta} {2} } \right)^{n-i+1}} \right] \mathcal{L}_{I_{}}^{(i)} (s)\notag \\
&+\sum_{i=0}^{n-1} \binom {n-1} {i} (-1)^{n-i} (n-i)!\left[ \pi (1-t_f) \lambda_b \int_{0}^{\infty} \frac {\left( v^{-\frac{\beta} {2} }\right)^{n-i} \mathrm{d} v} {\left(1+s v^{-\frac{\beta} {2} } \right)^{n-i+1}}  \right] \mathcal{L}_{I_{}}^{(i)} (s).
\end{align}
According to the definition of $y_n^{}$ and $s$, we have
\begin{align}
y_n^{}=a \sum_{i=0}^{n-1}  \frac {n-i} {n} l_{n-i}^{} y_i^{},
\end{align}
where $l_{i}=\left((1-t_f)\int_{0}^{\infty} \frac {\left( w^{-\frac{\beta} {2} }\right)^{i} \mathrm{d} w} {\left(1+ w^{-\frac{\beta} {2} } \right)^{i+1}} + t_f\int_{\tau^{-2 \backslash \beta}}^{\infty} \frac {\left( w^{-\frac{\beta} {2} }\right)^{i} \mathrm{d} w} {\left(1+ w^{-\frac{\beta} {2} } \right)^{i+1}} \right), i=\{1,2,\cdots, N-1 \}$ and $a=\pi \lambda_b {\|\mathbf{x}_1 \|}^2 \tau^{2 \backslash \beta}$. Note that $l_{i}$ can be expressed as the combination of the Gauss hypergeometric function and the Beta function, which is presented in Theorem \ref{MRTsuccessP}. Let $l_{0}= \big ( t_f \frac{  2\tau} {\beta-2} {}_2 F_1 \big [1,1-\frac{2} {\beta};2-\frac{2} {\beta};-\tau\ \big] +(1-t_f) \frac {2 \pi} {\beta} \csc \left(\frac{2 \pi} {\beta}  \big) \tau^{2 \backslash \beta} \right)  $ and we then have $y_0^{}=\mathcal{L}_{I_{}} (\tau {\|\mathbf{x}_1\|}^{\beta}))=\exp \left(- \pi \lambda_b  l_{0}  {\|\mathbf{x}_1\|}^2  \right)$.

To get the expression of $y_n$, we need to solve a series of linear equality. We then construct a Toeplitz matrix  as \cite{Changli14SmallCell} and after some manipulations, we obtain $P_{\textrm{s}}^{f,c}\left(t_f\right)$ as follows
\begin{align}
P_{\textrm{s}}^{f,c}(t_f)= \mathbb{E}_{\mathbf{x}_1} \left[ \left \|y_0^{} \sum_{i=0}^{N-1} \frac {1} {i!} a^i \mathbf{D}_{}^{i}  \right\|_1 \right],
\end{align}
where $\| \|_1$ is the $l_1$ induced matrix norm and the expression of $\mathbf{D}_{}$ is
 \begin{align}
\mathbf{D}_{}=\begin{bmatrix}
0 &  \\
l_{1} &0 \\
l_{2} & l_{1} &0  \\
\vdots &\vdots & &\ddots \\
l_{N}& l_{M-2}^{} &\cdots &l_{1} & 0
\end{bmatrix}.
\end{align}

According to the thinning of the PPP, the PDF of the distance of the closest $f$-cached BS to the typical user is $ 2 \pi t_f  \lambda_b \|\mathbf{x}_1\| e^{-\pi t_f \lambda_b {\|\mathbf{x}_1\|}^2}$. After taking expectation over $\mathbf{x}_1$ and utilizing the Taylor expansion, the STP is given by
\begin{align}
P_{\textrm{s}}^{f,c} \left(t_f\right)=\frac {t_f} {t_f+ l_{0}} \left \| \left[ \mathbf{I}- \left(\frac {\tau^{2 \backslash \beta}} {t_f+ l_{0}}\right)\  \mathbf{D}_{} \right]^{-1}\right \|_1.
\end{align}

 We can obtain $P_{\textrm{s}}^{b} , f \in \mathcal{F}_b$ similarly and we omit the details due to space limitation.

\vspace{-0.8em}

\subsection{Proof of \emph{Lemma} \ref{propertiesSTP}}\label{ProofSTPproperties}
Firstly, we prove the property 1.  Let $\mathbf{A} \triangleq \left[ \mathbf{I}- \left(\frac {\tau^{2 \backslash \beta}} {t_f+ l_{0}^{c,f}}\right)\  \mathbf{D}_{}^{c,f} \right]$ and we have $P_{\textrm{s}}^{f,c}(t_f)=\frac {t_f} {t_f+ l_{0}^{c,f}} \left \| \mathbf{A} ^{-1} \right \|_1$. We first derive the lower bound of $\left \| \mathbf{A} ^{-1} \right \|_1$. For any $\mathbf{x}$ and $\mathbf{y}$ satisfying $\mathbf{y} =\mathbf{A}^{-1} \mathbf{x}$, we have $\left \|  \mathbf{A}^{-1} \right \|_1 \geq \frac {\| \mathbf{y}\|_1}  {\| \mathbf{x}\|_1}$ due to the inequality ${\| \mathbf{y}\|_1}  \leq \left \|  \mathbf{A}^{-1} \right \|_1  {\| \mathbf{x}\|_1}$. Let $\mathbf{y}=[1,1,\cdots,1]^{T}$ and then we have
 \begin{align}\label{upperinvA}
 \left \|  \mathbf{A}^{-1} \right \|_1 \geq \frac {\| \mathbf{y}\|_1}  {\| \mathbf{x}\|_1} = \frac{N} {N-\frac {\tau^{2 \backslash \beta}} {t_f+ l_{0}^{c,f}} \sum_{i=1}^{N-1} (N-i) l_{i}^{c,f} }.
 \end{align}

We then derive the upper bound of $\left \| \mathbf{A} ^{-1} \right \|_1$. Noticing that $\mathbf{A} ^{-1} =\left(\mathbf{I}-\mathbf{A} \right)  \mathbf{A} ^{-1} +\mathbf{I}$ and using the triangle inequality, we have $\left\|  \mathbf{A}^{-1} \right\|_1 \leq   \left \| \left(\mathbf{I}-\mathbf{A} \right) \right \|_1 \left\|  \mathbf{A}^{-1} \right\|_1 +\left\|  \mathbf{I} \right\|_1$. Therefore, we obtain the upper bound of $\left \| \mathbf{A} ^{-1} \right \|_1$ as follows
\begin{align}\label{lowerinvA}
\left \| \mathbf{A} ^{-1} \right \|_1 \leq \frac{\left\|  \mathbf{I} \right\|_1} {1- \left \| \left(\mathbf{I}-\mathbf{A} \right) \right \|_1} =\frac {1} {1-\frac {\tau^{2 \backslash \beta}} {t_f+ l_{0}^{c,f}} \sum_{i=1}^{N-1} l_{i}^{c,f}}.
\end{align}
Note that $P_{\textrm{s}}^{f,c}(t_f)=\frac {t_f} {t_f+ l_{0}^{c,f}} \left \| \mathbf{A} ^{-1} \right \|_1$, we get the bounds of $P_{\textrm{s}}^{f,c}(t_f)$ as (\ref{STPupperlowerTf}) after substituting the expression of $l_{i}^{c,f}$ and $l_{0}^{c,f}$ in Theorem \ref{MRTsuccessP}. It can be shown that $l_{i+1,{}}^{c,f} \leq l_{i}^{c,f}$ for $i \in \mathbb{N}$ and  $l_{0}^{c,f}=\sum_{i=1}^{\infty} l_{i}^{c,f}, \forall t_f \in [0,1]$. Moreover, $l_{i}^{c,f}$ decreases with  $t_f$  for $i \in \mathbb{N}$. Therefore, after carefully checking the properties of $\nu_A$, $\mu_A$, $\nu_B$ and $\mu_B$, we obtain property 1.

Secondly, we prove property 2. We define $\mathbf{B}  \triangleq \left( t_f+ l_{0}^{c,f} \right) \mathbf{A}$ and then we have $P_{\textrm{s}}^{f,c}\left(t_f\right)= \left \|  t_f \mathbf{B}^{-1} \right \|_1$.
Furthermore, the derivative of $\mathbf{B}$ w.r.t. $t_f$ is a lower triangular Toeplitz matrix and
 \begin{align}\label{expressionK}
\frac {\partial  \mathbf{B} } {\partial t_f} =\begin{bmatrix}
1-k_0&  \\
k_1  &1-k_0 \\
k_2 & k_1 &1-k_0 \\
\vdots &\vdots & &\ddots \\
k_{N-1} &k_{N-2} &\cdots &k_1  & 1-k_0
\end{bmatrix},
\end{align}
where $k_0$ and $k_i, i \in \{1,2,\cdots,N-1\}$ are given by
\begin{align}
&k_0= \frac {2 \pi} {\beta} \csc \left(\frac{2 \pi} {\beta} \right) \tau^{2 \backslash \beta} -\frac{  2\tau} {\beta-2} {}_2 F_1 \left[1,1-\frac{2} {\beta};2-\frac{2} {\beta};-\tau\ \right], \\
&k_i= \frac{2 \tau^{2 \backslash \beta} } {\beta} B(\frac{2} {\beta}+1,i-\frac{2} {\beta}) - \frac{2 \tau^{i}} {i \beta-2} {}_2 F_1 \left[i+1,i-\frac{2} {\beta};i+1-\frac{2} {\beta};-\tau\ \right],1 \leq i \leq N-1.
\end{align}

Then we derive the derivative of  $P_{\textrm{s}}^{f,c}(t_f)$ w.r.t. $t_f$ as follows
\begin{align}
\frac{\partial P_{\textrm{s}}^{f,c}(t_f) } {\partial t_f} &=\left \|  \frac {\partial (t_f \mathbf{B}^{-1})} {\partial t_f} \right \|_1 \notag \\
&=\left \|\mathbf{B}^{-1} - t_f  \mathbf{B}^{-1}  \frac {\partial  \mathbf{B}} {\partial t_f} \mathbf{B}^{-1} \right \|_1  \notag \\
&\geq \frac {1} {t_f} \left( \left \|t_f \mathbf{B}^{-1} \right \|_1 - \left \| t_f  \mathbf{B}^{-1}  \frac {\partial  \mathbf{B}} {\partial t_f} \mathbf{B}^{-1} t_f \right \|_1  \right) \notag \\
&\geq \frac {1} {t_f} \left \|t_f \mathbf{B}^{-1} \right \|_1 \left(1- \left \| \frac {\partial  \mathbf{B} } {\partial t_f} \right \|_1  \left \|t_f \mathbf{B}^{-1} \right \|_1 \right) \notag \\
&=\frac {1} {t_f} \left \|t_f \mathbf{B}^{-1} \right \|_1 \left(1- \left(1-k_0+\sum_{i=1}^{N-1} k_i  \right) \left \|t_f \mathbf{B}^{-1} \right \|_1 \right).
\end{align}
Note that $1-k_0+\sum_{i=1}^{N-1} k_i \leq 1-k_0+\sum_{i=1}^{\infty} k_i=1$, therefore we have
\begin{align}
\frac{\partial P_{\textrm{s}}^{f,c}(t_f) } {\partial t_f}&\geq \frac {1} {t_f} \left \|t_f \mathbf{B}^{-1} \right \|_1 \left(1-   \left \|t_f \mathbf{B}^{-1} \right \|_1 \right) =\frac {1} {t_f} P_{\textrm{s}}^{f,c}(t_f) \left(1-   P_{\textrm{s}}^{f,c}(t_f)  \right).
\end{align}
Note that $P_{\textrm{s}}^{f,c}(t_f) \in [0,1]$, we then finish the proof of property 2 because $\frac{\partial P_{\textrm{s}}^{f,c}(t_f) } {\partial t_f} \geq 0$

Finally, we prove property 3.  Note that $P_{\textrm{s}}^{f,c}(t_f)=P_{\textrm{s}}^{b}$ if and only if $t_f=1$. According to property 2, we have $P_{\textrm{s}}^{f,c}(t_f) \leq P_{\textrm{s}}^{b}$. According to property 1, the upper bound and the lower bound of $P_{\textrm{s}}^{f,c}(t_f)$ are both $0$ if $t_f=0$, therefore we have $P_{\textrm{s}}^{f,c}(t_f)=0$ if $t_f=0$.

\vspace{-0.4em}

\subsection{Proof of Theorem \ref{optimalFc}}\label{ProofofFcallocation}
To prove property 1 is equivalent to prove that $B \leq |\mathcal{F}_b^{*}| \leq F-C$ for optimizing $P_{\textrm{s}}\left( \mathcal{F}_c,\mathbf{t} \right)$.

We first prove that $|\mathcal{F}_b^{*}| \leq F-C$. Suppose that there exists optimal $\left(\mathcal{F}_b^{*},\mathbf{t}^{*} \right)$ to problem \ref{problemMRTorigin}  satisfying $|\mathcal{F}_b| > F-C$. Now we construct a feasible solution $\left(\mathcal{F}_b^{'}, \mathbf{t}^{'} \right)$ to problem \ref{problemMRTorigin}, where $\mathcal{F}_b^{'}$ is the set of the most popular $F-C$ files of $\mathcal{F}_b^{*}$, the elements of $\mathbf{t}^{'}$ are same as $\mathbf{t}^{*}$  if $f \in \{ \mathcal{F} \setminus \mathcal{F}_b^{*} \}$ and are one if $f \in \{ \mathcal{F}_b^{*} \setminus  \mathcal{F}_b^{'}\}$. Due to the fact that $F_c^{'} \leq C$, $\left(\mathcal{F}_b^{'}, \mathbf{t}^{'} \right)$ is a feasible solution satisfying the constraints.  Note that $P_{\textrm{s}}^{f,c} (t_f)=P_{\textrm{s}}^{b}$ if $t_f=1$ , we then have $P_{\textrm{s}}(\mathbf{t}^{'},\mathcal{F} \setminus  \mathcal{F}_b^{'})-P_{\textrm{s}}(\mathbf{t}^{*},\mathcal{F} \setminus  \mathcal{F}_b^{*})=\sum_{f \in \mathcal{F}_b^* \setminus \mathcal{F}_b^{'}} q_f    \left(P_{\textrm{s}}^{f,c} (1)-\sum_{k=1}^{F_b} \mathbb{P} \left(F_{b,1}^r=k \right)\frac{B} {\max\left(k,B\right)} P_{\textrm{s}}^{b} \right)>0$, which contradicts with the optimality of  $\left(\mathcal{F}_b^{*},\mathbf{t}^{*} \right)$. Therefore, we prove $\left|\mathcal{F}_b^{*} \right| \leq F-C$.

We then prove that $|\mathcal{F}_b^{*}| \geq B$. Suppose that there exist optimal $\left(\mathcal{F}_b^{*},\mathbf{t}^{*} \right)$ to problem \ref{problemMRTorigin}  satisfying $|\mathcal{F}_b| < B$. Now we construct a feasible solution $\left(\mathcal{F}_b^{'}, \mathbf{t}^{'} \right)$ to problem \ref{problemMRTorigin}, where $\mathcal{F}_b^{'}$ is the combining of  $\mathcal{F}_b^{*}$ and any $B-|\mathcal{F}_b^{*}|$ files in $\left \{ \mathcal{F} \setminus \mathcal{F}_b^{*} \right \}$, the elements of $\mathbf{t}^{'}$ are same as $\mathbf{t}^{*}$  if $f \in \{ \mathcal{F} \setminus \mathcal{F}_b^{'} \}$.  When $\left|\mathcal{F}_b \right| \leq B$,
$P_{\textrm{s}}(\mathcal{F}_c,\mathbf{t})= \sum_{f \in \mathcal{F}_c} q_f  P_{\textrm{s}}^{f,c} (t_f) + \sum_{f \in \mathcal{F} \setminus \mathcal{F}_c} q_f P_{\textrm{s}}^{b}$. Note that $P_{\textrm{s}}^{b} \geq P_{\textrm{s}}^{f,c} (t_f)$ , we then have $P_{\textrm{s}}(\mathbf{t}^{'},\mathcal{F} \setminus  \mathcal{F}_b^{'})-p_{\textrm{s}}(\mathbf{t}^{*},\mathcal{F} \setminus  \mathcal{F}_b^{*})=\sum_{f \in \mathcal{F}_b^{,} \setminus \mathcal{F}_b^{*}}   q_f  \left( P_{\textrm{s}}^{b} -P_{\textrm{s}}^{f,c} (t_f) \right) \geq 0$, which contradicts with the optimality of  $\left( \mathcal{F}_b^{*},\mathbf{t}^{*} \right)$. Therefore, we prove $\left|\mathcal{F}_b^{*} \right| \geq B$.

\vspace{-0.8em}

\subsection{Proof of Lemma \ref{upperboundproperty} }\label{Proofpropertyupper}
To prove Lemma \ref{upperboundproperty} is equivalent to prove $P_{\textrm{s}}^{u,f,c} (t_f)$ is increasing w.r.t. $t_f$ for $f \in \mathcal{F}_c$.  When $N=1$, we have $P_{\textrm{s}}^{u,f,c} (t_f)= \frac { q_f t_f} {\zeta_1 (\alpha \tau) t_f +\zeta_2 (\alpha \tau)}$ and it is increasing w.r.t. $t_f$. We then consider the scenario when $N \geq 2$.  According to the proof of the upper bound, $P_{\textrm{s}}^{u,f,c} (t_f)= 1-\mathbb{E}_{I_{}{\| \mathbf{x}_1\|}^{\beta}} \left[ \left( 1-\exp  \left( -\alpha \tau I_{}{\| \mathbf{x}_1\|}^{\beta} \right) \right)^{N} \right] $. After taking expectation over $I$, we have
\begin{align}
P_{\textrm{s}}^{u,f,c} (t_f)=1-\mathbb{E}_{{\| \mathbf{x}_1\|}^{\beta}} \left[ \left(1- \exp \left( \pi \lambda_b  \left( \left(1-\zeta_1 (\alpha \tau) \right) t_f -\zeta_2 (\alpha \tau) \right)  {\|\mathbf{x}_1\|}^2  \right)  \right)^{N} \right],
\end{align}
where $\zeta_1 (\alpha \tau)$ and $\zeta_2 (\alpha \tau)$ are given in (\ref{expressionzeta1}) and (\ref{expressionzeta2}).

Let $U\left(t_f | \mathbf{x}_1\right) \triangleq  \left(1- \exp \left(\pi \lambda_b  \left( \left(1-\zeta_1 (\alpha \tau)\right) t_f -\zeta_2 (\alpha \tau) \right)  {\|\mathbf{x}_1\|}^2  \right)  \right)$ and we have
\begin{align}
&\frac{\partial U^{N} \left(t_f | \mathbf{x}_1\right) } {\partial t_f} =N  \left( \pi \lambda_b {\|\mathbf{x}_1\|}^2 \left(\zeta_1 (\alpha \tau)-1\right) \right) U^{N-1} \left(t_f | \mathbf{x}_1\right)  \exp \left( \pi \lambda_b  \left( \left( 1-\zeta_1 (\alpha \tau) \right) t_f -\zeta_2 (\alpha \tau) \right)  {\|\mathbf{x}_1\|}^2  \right).  \notag
\end{align}
Note that $\zeta_1 (\alpha \tau) <1$ and $\zeta_2 (\alpha \tau)>0$, we have $\frac{\partial U^{N} \left(t_f | \mathbf{x}_1\right) } {\partial t_f}<0$ and $U^{N} \left(t_f | \mathbf{x}_1\right) $ is   a decreasing function of $t_f$ for any $ \mathbf{x}_1$. Therefore, $P_{\textrm{s}}^{u,f,c} (t_f)=1-\mathbb{E}_{{\| \mathbf{x}_1\|}^{\beta}} \left[U^{N} \left(t_f | \mathbf{x}_1 \right)\right]$ is an increasing function of $t_f$. This is because $\mathbb{E}_{{\| \mathbf{x}_1\|}^{\beta}} \left[U^{N} \left(t_f | \mathbf{x}_1 \right)\right]$ can be interpreted as a combination of a series of $U^{N}\left(t_f | \mathbf{x}_1 \right)$ with different $\mathbf{x}_1$ and the monotonicity holds.

\vspace{-0.8em}

\subsection{Proof of Theorem \ref{aymoptFc} }\label{ProofofasymTFc}
To prove $\mathcal{F}_c^{*}=\{B+1,B+2,\cdots,F\} $  is equivalent to prove $\mathcal{F}_b^{*}=\{1,2,\cdots,B\}$.

Firstly we prove that $|\mathcal{F}_b^*| \leq B$. Suppose that there exist optimal $\left(\mathcal{F}_b^{*},\mathbf{t}^{*} \right)$ to problem \ref{problemMRTorigin}  satisfying $|\mathcal{F}_b^{*}| > B$ and then we have $\max \left( F_b^{*},B \right)=F_b$. Now we construct a feasible solution $\left(\mathcal{F}_b^{'}, \mathbf{t}^{'} \right)$ to problem \ref{ProblemMRTasymT}, where $\mathcal{F}_b^{'}$ is the set of the most popular $B$ files of $\mathcal{F}_b^{*}$, the elements of $\mathbf{t}^{'}$ are same as $\mathbf{t}^{*}$  if $f \in \{ \mathcal{F} \setminus \mathcal{F}_b^{*} \}$ and are zero if $f \in \{ \mathcal{F}_b^{*} \setminus  \mathcal{F}_b^{'}\}$.  Note that $P_{\textrm{s}}^{u,f,c} \left(t_f \right)=0$ if and only if $t_f=0$ , we then have $R_{\textrm{s},\infty}^{u} (\mathbf{t}^{'},\mathcal{F} \setminus  \mathcal{F}_b^{'})-R_{\textrm{s},\infty}^{u}(\mathbf{t}^{*},\mathcal{F} \setminus  \mathcal{F}_b^{*})= \lambda_b \log_2(1+\tau) \left( \sum_{f \in \mathcal{F}_b^{'}} \frac {q_f P_{\textrm{s}}^{u,b}  } {B }-\sum_{f \in \mathcal{F}_b^{*}}  \frac { q_f P_{\textrm{s}}^{u,b}  } {F_b } \right)  B>0$ because $P_{\textrm{s}}^{u,b}>0$, which contradicts with the optimality of  $\left( \mathcal{F}_b^{*},\mathbf{t}^{*} \right)$. Therefore, we can prove $\left|\mathcal{F}_b^{*} \right| \leq B$.

Secondly we prove that $|\mathcal{F}_b^*| \geq B$. Note that $P_{\textrm{s}}^{u,f,c} \left(t_f \right)$ increases w.r.t. $t_f$ and $P_{\textrm{s}}^{u,f,c}\left(t_f\right) \leq P_{\textrm{s}}^{u,b}$, therefore we can easily prove $\left|\mathcal{F}_b^{*} \right| \geq B$ similarly as Appendix \ref{ProofofFcallocation}.

Combining the results $\left|\mathcal{F}_b^{*} \right| \leq B$ and $\left|\mathcal{F}_b^{*} \right| \geq B$, we have $\left|\mathcal{F}_b^{*} \right| = B$. Finally we prove $\mathcal{F}_b^{*}=\{1,2,\cdots,B\}$.  Suppose that there exists optimal $\left(\mathcal{F}_b^{*},\mathbf{t}^{*} \right)$ to problem \ref{problemMRTorigin}  satisfying $|\mathcal{F}_b^{*}|=B$ and $\mathcal{F}_b \neq \{1,2,\cdots\,B\}$.  We construct a feasible solution $\left(\mathcal{F}_b^{'}, \mathbf{t}^{'} \right) $ to problem \ref{ProblemMRTasymT}, where $\mathcal{F}_b^{'} = \{1,2,\cdots\,B\}$, the elements of $\mathbf{t}^{'}$ are same as $\mathbf{t}^{*}$  if $f \in \{  \mathcal{F} \setminus \left( \mathcal{F}_b^{*} \bigcup \mathcal{F}_b^{'} \right) \} $ and are same as $\mathbf{t}^{*}$ in order for the rest files, i.e.,   $t^{'}_{f_n}=t^{*}_{f_m}, {f_n} \in \left \{ \mathcal{F}_b^{*} \bigcap \mathcal{F}_c^{'} \right \}, {f_m} \in \left \{ \mathcal{F}_b^{'} \bigcap \mathcal{F}_c^{*} \right \}, n=m$.  Note that $n$ and $m$ denote the order of the file in $\left \{ \mathcal{F}_b^{*} \bigcap \mathcal{F}_c^{'} \right \}$ and $\left \{ \mathcal{F}_b^{'} \bigcap \mathcal{F}_c^{*} \right \}$. We then have $R_{\textrm{s},\infty}^{u} (\mathbf{t}^{'},\mathcal{F} \setminus  \mathcal{F}_b^{'})-R_{\textrm{s},\infty}^{u}(\mathbf{t}^{*},\mathcal{F} \setminus  \mathcal{F}_b^{*})= \lambda_b \log_2(1+\tau)  \sum_{f_m \in \left \{ \mathcal{F}_b^{'} \bigcap \mathcal{F}_c^{*} \right \} } q_{f_m} \left( P_{\textrm{s}}^{u,b}- P_{\textrm{s}}^{c,b}\left( t_{f_m}^{*} \right)\right)  -\lambda_b \log_2(1+\tau) \sum_{f_n \in \left \{ \mathcal{F}_b^{*} \bigcap \mathcal{F}_c^{'} \right \} } q_{f_n} \left( P_{\textrm{s}}^{u,b}- P_{\textrm{s}}^{c,b}\left( t_{f_n}^{'} \right)\right)$. Note that  $P_{\textrm{s}}^{u,b}- P_{\textrm{s}}^{c,b}\left( t_{f_m}^{*} \right)= P_{\textrm{s}}^{u,b}- P_{\textrm{s}}^{c,b}\left( t_{f_n}^{'} \right)$ if $m=n$, we then obtain $R_{\textrm{s},\infty}^{u} (\mathbf{t}^{'},\mathcal{F} \setminus  \mathcal{F}_b^{'})-R_{\textrm{s},\infty}^{u}(\mathbf{t}^{*},\mathcal{F} \setminus  \mathcal{F}_b^{*})= \lambda_b \log_2(1+\tau)  \sum_{m \in \left \{ 1,2, \cdots, \left | \left \{ \mathcal{F}_b^{'} \bigcap \mathcal{F}_c^{*} \right \} \right| \right \} } \big( q_{f_m} -q_{f_n} \big) \big( P_{\textrm{s}}^{u,b}- P_{\textrm{s}}^{c,b} \big( t_{f_m}^{*} \big) \big) >0$. The last inequality is because $q_{f_m}>q_{f_n}$ if $n=m$ and $P_{\textrm{s}}^{u,b}> P_{\textrm{s}}^{c,b}\left( t_{f_m}^{*} \right)$.

\end{appendix}

\bibliographystyle{IEEEtran}
\bibliography{reference}
\end{document}